%% file: main.tex
\documentclass[a4paper,UKenglish,cleveref, autoref, thm-restate]{lipics-v2021}

\pdfoutput=1 
\hideLIPIcs  

\title{ \bf Fast and Secure Decentralized Optimistic Rollups Using
  Setchain}

 \author{Margarita Capretto}{IMDEA Software Institute, Spain \and Universidad Politécnica de Madrid}{margarita.capretto@imdea.org}{https://orcid.org/0000-0003-2329-3769}{}
 \author{Martín Ceresa}{IMDEA Software Institute,
   Spain}{martin.ceresa@imdea.org@imdea.org}{https://orcid.org/0000-0003-4691-5831}{}
 \author{Antonio {Fernández Anta}}{IMDEA Networks Institute, Spain}
 {antonio.fernandez@imdea.org}{https://orcid.org/0000-0001-6501-2377}{}
 \author{Pedro Moreno-Sánchez}{IMDEA Software Institute,
   Spain \and VISA Research, Spain}{pedro.moreno@imdea.org@imdea.org}{https://orcid.org/0000-0003-2315-7839}{}
\author{César Sánchez}{IMDEA Software Institute,
   Spain}{cesar.sanchez@imdea.org@imdea.org}{https://orcid.org/0000-0003-3927-4773}{}

\authorrunning{M. Capretto, M. Ceresa, A. {Fernandez Anta},
  P. Moreno-Sanchez and C. Sanchez} 

\Copyright{Jane Open Access and Joan R. Public} 

\ccsdesc[100]{Information systems$\rightarrow$Distributed systems
  organizing principles; Systems organizing principles: Blockchain
  technology; Software and its engineering$\rightarrow$Distributed systems
  organizing principles; Distributed ledgers}

\keywords{Blockchain, Layer Two, Data Availability Committee} 

\category{} 

\relatedversion{} 

\nolinenumbers 

\EventEditors{Rainer B\"ohme and Lucianna Kiffer}
\EventNoEds{2}
\EventLongTitle{Advances in Financial Technologies (AFT 2024)}
\EventShortTitle{AFT 2024}
\EventAcronym{AFT}
\EventYear{2024}
\EventDate{September 23--25, 2024}
\EventLocation{Vienna, Austria}
\EventLogo{}
\SeriesVolume{VV}
\ArticleNo{NN}

\input{preamble}

\input{macros}

\input{solidity-highlighting.tex}

\begin{document}

\maketitle
\input{Abstract}

\input{Introduction}

\input{Preliminaries}
\input{API}
\input{SeqDecentralized}
\input{Incentives}
\input{ThreatModel}
\input{EmpiricalEvaluation}
\input{RelatedWork}
\input{Conclusion}

\vfill
\clearpage

\bibliography{bibfile}
\vfill

\appendix
\pagebreak
\input{Appendix/SeqDC}
\vfill
\pagebreak
\input{Appendix/SetchainImpl}

\vfill
\pagebreak
\input{Appendix/htlc}
\vfill
\pagebreak
\input{Appendix/Challenges-alg}
\vfill
\pagebreak
\input{Appendix/ProofIncentives}
\vfill
\pagebreak
\input{Appendix/DecentralizationL2Eth}
\vfill
\pagebreak
\input{Appendix/ExtendedRelatedWork}
\end{document}

%% file: preamble.tex
\usepackage[utf8]{inputenc}
\usepackage[ruled]{algorithm}
\usepackage[noend]{algpseudocode}
\usepackage{amsmath}
\usepackage{amsfonts}
\usepackage{amsthm}
\usepackage{mathtools}
\usepackage[disable]{todonotes}
\usepackage[ligature,reserved]{semantic}
\usepackage{xspace}
\usepackage{csquotes}

\usepackage{scalerel}
\usepackage{tikz}
\usepackage{adjustbox}

\usetikzlibrary{svg.path}

\definecolor{orcidlogocol}{HTML}{A6CE39}
\tikzset{
    orcidlogo/.pic={
        \fill[orcidlogocol] svg{M256,128c0,70.7-57.3,128-128,128C57.3,256,0,198.7,0,128C0,57.3,57.3,0,128,0C198.7,0,256,57.3,256,128z};
        \fill[white] svg{M86.3,186.2H70.9V79.1h15.4v48.4V186.2z}
        svg{M108.9,79.1h41.6c39.6,0,57,28.3,57,53.6c0,27.5-21.5,53.6-56.8,53.6h-41.8V79.1z M124.3,172.4h24.5c34.9,0,42.9-26.5,42.9-39.7c0-21.5-13.7-39.7-43.7-39.7h-23.7V172.4z}
        svg{M88.7,56.8c0,5.5-4.5,10.1-10.1,10.1c-5.6,0-10.1-4.6-10.1-10.1c0-5.6,4.5-10.1,10.1-10.1C84.2,46.7,88.7,51.3,88.7,56.8z};
    }
}
\newcommand\orcidicon[1]{\href{https://orcid.org/#1}{\mbox{\scalerel*{
                \begin{tikzpicture}[yscale=-1,transform shape]
                \pic{orcidlogo};
                \end{tikzpicture}
            }{|}}}}

\algblockdefx[Upon]{Upon}{EndUpon}%
[1]{{\bf upon} (#1) {\bf do}}%
{{\bf end upon}}

\algblockdefx[When]{When}{EndWhen}%
[1]{{\bf when} (#1) {\bf do}}%
{{\bf end when}}

\algblockdefx[Receive]{Receive}{EndReceive}%
[1]{{\bf receive} (#1) {\bf do}}%
{{\bf end receive}}

\makeatletter
\ifthenelse{\equal{\ALG@noend}{t}}%
  {\algtext*{EndUpon}}
  {}%
\makeatother

\makeatletter
\ifthenelse{\equal{\ALG@noend}{t}}%
  {\algtext*{EndWhen}}
  {}%
\makeatother

\makeatletter
\ifthenelse{\equal{\ALG@noend}{t}}%
  {\algtext*{EndReceive}}
  {}%
\makeatother

\mathlig{<-}{\leftarrow}
\mathlig{==}{\equiv}
\mathlig{++}{\mathop{+\!\!+}}
\mathlig{>>}{\rangle}
\mathlig{~}{\sim}
\reservestyle{\variables}{\text}
\variables{epoch,current,history,theset[the\_set],proposal,valid}

\reservestyle{\setops}{\textsc}
\setops{add,get,Init,Add,Get,BAdd,EpochInc,Broadcast,Deliver,GetEpoch,Propose,Inform,SetDeliver,
  SuperSetDeliver, DoStuff[Start],Consensus}
\setops{translate, Translate}

\reservestyle{\structs}{\text}
\structs{DPO,BAB,BRB,SBC,DSO}
\structs{ArrangerApi[Arranger], SequencerApi[Sequencer], DCsApi[DC
  members], DCApi[DC member]}

\reservestyle{\stmt}{\textbf}
\stmt{call,ack,drop,return,assert,wait, post, inform, trigger, logger}

\reservestyle{\mathfunc}{\mathsf}
\mathfunc{send,receive,sendbrb[send\_brb],sendsbc[send\_sbc], BLSAggregate[Aggr]}

\reservestyle{\messages}{\texttt}
\messages{madd[add],mepochinc[epinc]}

\reservestyle{\api}{\texttt}
\api{apiBAdd,apiAdd,apiGet,apiEpochInc,apiTheSet[{the\_set}],apiHistory}

\reservestyle{\servers}{\texttt}
\servers{S,C,B,M,DAC}

\reservestyle{\schain}{\texttt}
\schain{history,epoch,theset[{the\_set}],add,get,epochinc[{epoch\_inc}],getepoch[{get\_epoch}],tobroadcast[{to\_broadcast}], knowledge, pending, sent, received,proposal,Valid}
\schain{add, translate, allTxs, pendingTxs, maxTime, maxSize,
  signReq, signResp, hashes, setchain, localsetchain[Setchain],
  invalidId, invalidHash, notContacted[left], signatures} 

\reservestyle{\setvars}{\textit}
\setvars{prop,propset,h,tx,id,sigs,hash, batchId[id]}

\reservestyle{\event}{\texttt}
\event{myturn[my\_turn], newepoch[new\_epoch],added, timetopost[time\_to\_post],tobatch[to\_batch], nextBatch[next\_batch\_id]}

\newcommand\imarga[1]{\todo[inline]{Marga: {#1}}}

\newtheorem{property}{Property}
\newtheorem*{property*}{Property}
\usepackage{xspace}
\usepackage{url}
\usepackage{paralist}

\usepackage{environ}
\usepackage{wrapfig}

\usepackage{bibunits}
\usepackage{hyperref} 
\defaultbibliography{bibfile}
\defaultbibliographystyle{plainurl}


%% file: macros.tex
\newcommand{\EPOCH}{\<epoch>}

\newcommand{\HISTORY}{\<history>}
\newcommand{\THESET}{\<theset>}

\newcommand{\arranger}{Arranger\xspace}
\newcommand{\hash}{\mathcal{H}}

\newcommand{\repeattheorem}[1]{%
  \begingroup
  \renewcommand{\thetheorem}{\ref{#1}}%
  \expandafter\expandafter\expandafter\theorem
  \csname reptheorem@#1\endcsname
  \endtheorem
  \endgroup
}

\NewEnviron{reptheorem}[1]{%
  \global\expandafter\xdef\csname reptheorem@#1\endcsname{%
    \unexpanded\expandafter{\BODY}%
  }%
  \expandafter\theorem\BODY\unskip\label{#1}\endtheorem
}

\NewEnviron{replemma}[1]{%
  \global\expandafter\xdef\csname replemma@#1\endcsname{%
    \unexpanded\expandafter{\BODY}%
  }%
  \expandafter\lemma\BODY\unskip\label{#1}\endlemma
}
\newcommand{\repeatatlemma}[1]{%
  \begingroup
  \renewcommand{\thelemma}{\ref{#1}}%
  \expandafter\expandafter\expandafter\lemma
  \csname replemma@#1\endcsname
  \endlemma
  \endgroup
}

\NewEnviron{repprop}[1]{%
  \global\expandafter\xdef\csname repprop@#1\endcsname{%
    \unexpanded\expandafter{\BODY}%
  }%
  \expandafter\proposition\BODY\unskip\label{#1}\endproposition
}
\newcommand{\repeatatproposition}[1]{%
  \begingroup
  \renewcommand{\theproposition}{\ref{#1}}%
  \expandafter\expandafter\expandafter\proposition
  \csname repprop@#1\endcsname
  \endlemma
  \endgroup
}

\newcommand{\PROP}[1]{\textbf{\textit{#1}}\xspace}
\newcommand{\PrAvailability}{\PROP{Availability}}
\newcommand{\PrTermination}{\PROP{Termination}}
\newcommand{\PrValidity}{\PROP{Validity}}
\newcommand{\PrIntegrity}{\PROP{Integrity}}
\newcommand{\PROPsub}[2]{\ensuremath{\bf \textit{\textbf{#1}}_{#2}}\xspace}
\newcommand{\PrIntegrityOne}{\PROPsub{Integrity}{1}}
\newcommand{\PrIntegrityTwo}{\PROPsub{Integrity}{2}}
\newcommand{\PrUniqueBatch}{\PROP{Unique Batch}}
\newcommand{\PrCertified}{\PROP{Certified}}

\newcommand{\Ctt}{\texttt{C}\xspace}
\newcommand{\Stt}{\texttt{S}\xspace}
\newcommand{\Att}{\texttt{A}\xspace}

 \newcommand{\EthTPS}{$12$\xspace}

\newcommand{\LdosTPS}{$100$\xspace}

\newcommand{\KWD}[1]{\ensuremath{\mathit{#1}}\xspace}
\newcommand{\CC}[1]{\ensuremath{cc_{\KWD{#1}}}\xspace}
\newcommand{\SC}[1]{\ensuremath{sc_{\KWD{#1}}}\xspace}
\newcommand{\SR}[1]{\ensuremath{sr_{\KWD{#1}}}\xspace}
\newcommand{\CR}[1]{\ensuremath{cr_{\KWD{#1}}}\xspace}

\newcommand{\confirmedBatches}{\KWD{confirmedBatches}}
\newcommand{\proposedBatches}{\KWD{proposedBatches}}
\newcommand{\id}{\KWD{id}}
\newcommand{\merkleRoot}{\KWD{merkleRoot}}
\newcommand{\sigs}{\KWD{sigs}}
\newcommand{\staker}{\KWD{staker}}
\newcommand{\stakers}{\KWD{stakers}}
\newcommand{\val}{\KWD{value}}
\newcommand{\sender}{\KWD{sender}}
\newcommand{\challengeId}{\KWD{challengeId}}
\newcommand{\canBeChallenged}{\KWD{canBeChallenged}}
\newcommand{\challengedNodes}{\KWD{challengedNodes}}
\newcommand{\moreChallenges}{\KWD{moreChallenges}}
\newcommand{\aggPKs}{\KWD{aggPKs}}

%% file: solidity-highlighting.tex
\usepackage{listings, xcolor}

\definecolor{verylightgray}{rgb}{.97,.97,.97}

\lstdefinelanguage{Solidity}{
	keywords=[1]{anonymous, assembly, assert, balance, break, call, callcode, case, catch, class, constant, continue, constructor, contract, debugger, default, delegatecall, delete, do, else, emit, event, experimental, export, external, false, finally, for, function, gas, if, implements, import, in, indexed, instanceof, interface, internal, is, length, library, log0, log1, log2, log3, log4, memory, modifier, new, payable, pragma, private, protected, public, pure, push, require, return, returns, revert, selfdestruct, send, solidity, storage, struct, suicide, super, switch, then, this, throw, transfer, true, try, typeof, using, value, view, while, with, addmod, ecrecover, keccak256, mulmod, ripemd160, sha256, sha3}, 
	keywordstyle=[1]\color{blue}\bfseries,
	keywords=[2]{set,address, bool, byte, bytes, bytes1, bytes2, bytes3, bytes4, bytes5, bytes6, bytes7, bytes8, bytes9, bytes10, bytes11, bytes12, bytes13, bytes14, bytes15, bytes16, bytes17, bytes18, bytes19, bytes20, bytes21, bytes22, bytes23, bytes24, bytes25, bytes26, bytes27, bytes28, bytes29, bytes30, bytes31, bytes32, enum, int, int8, int16, int24, int32, int40, int48, int56, int64, int72, int80, int88, int96, int104, int112, int120, int128, int136, int144, int152, int160, int168, int176, int184, int192, int200, int208, int216, int224, int232, int240, int248, int256, mapping, string, elem, uint, uint8, uint16, uint24, uint32, uint40, uint48, uint56, uint64, uint72, uint80, uint88, uint96, uint104, uint112, uint120, uint128, uint136, uint144, uint152, uint160, uint168, uint176, uint184, uint192, uint200, uint208, uint216, uint224, uint232, uint240, uint248, uint256, var, void, ether, finney, szabo, wei, days, hours, minutes, seconds, weeks, years},	
	keywordstyle=[2]\color{teal}\bfseries,
	keywords=[3]{block, blockhash, coinbase, difficulty, gaslimit, number, timestamp, msg, gas, sender, sig, value, now, tx, gasprice, origin, add, epochinc, get, setminus, emptyset},	
	keywordstyle=[3]\color{violet}\bfseries,
	identifierstyle=\color{black},
	sensitive=false,
	comment=[l]{//},
	morecomment=[s]{/*}{*/},
	commentstyle=\color{gray}\ttfamily,
	stringstyle=\color{red}\ttfamily,
	morestring=[b]',
	morestring=[b]"
}

%% file: Abstract.tex
\begin{abstract}
  %
  %
  Modern blockchains face a scalability challenge due to the intrinsic
  throughput limitations of consensus protocols.
  Layer 2 optimistic rollups (L2) are a faster alternative that offer
  the same interface in terms of smart contract development and user
  interaction.
  Optimistic rollups perform most computations offchain and make light
  use of an underlying blockchain~(L1) to guarantee correct behavior,
  implementing a cheaper blockchain on a blockchain solution.
  With optimistic rollups, a \emph{sequencer} calculates offchain
  batches of L2 transactions and commits batches (compressed or
  hashed) to the L1 blockchain.
  %
  The use of hashes requires a data service to translate hashes into their
  corresponding batches.
  Current L2 implementations consist of a centralized sequencer
  (central authority) and an optional data availability committee
  (DAC).
  

  In this paper, we propose a decentralized L2 optimistic rollup based
  on \emph{Setchain}, a decentralized Byzantine-tolerant
  implementation of sets.
  %
  The main contribution is a \emph{fully decentralized ``arranger''}
  where arrangers are a formal definition combining sequencers and
  DACs.
  %
  %
  We prove our implementation correct and show empirical evidence
  that our solution scales.
  A final contribution is a system of incentives (payments) for
  servers that implement the sequencer and data availability committee
  protocols correctly, and a fraud-proof mechanism to detect
  violations of the protocol.
\end{abstract}


%% file: Introduction.tex
\section{Introduction}
\label{sec:intro}

%
%
Layer 2 (L2) rollups provide a faster alternative to blockchains, like
Ethereum~\cite{wood2014ethereum}, while offering the same interface in
terms of smart contract programming and user interaction.
%
%
Rollups perform as much computation as possible offchain with minimal
blockchain interaction, in terms of the number and size of
invocations, while guaranteeing a correct and trusted operation.
%
%
L2 solutions work in two phases: users inject transaction requests
through a \emph{sequencer}, and then, their effects are computed
\emph{offchain} and published in L1.
The sequencer orders the received transactions, packs them into
batches, compresses batches and injects the result in L1.
%
State Transition Functions~(SFTs) is the set of independent parties in
charge of computing transition effects offchain and publishing the
resulting effects in L1 encoded as L2 blocks.

There are two main categories of rollups:
\begin{compactitem}
\item \emph{ZK-Rollups:} STFs provide Zero-Knowledge
  proofs used to verify in L1 the correctness of the effects
  claimed.
\item \emph{Optimistic Rollups:} L2 blocks posted are
optimistically assumed to be correct
  delegating block correctness
  on a fraud-proof mechanism.
\end{compactitem}

Optimistic rollups include an arbitration process to solve disputes.
STFs place a stake along their fresh L2 block when posting the block
to L1.
Since blocks can be incorrect, there is a fixed interval during which
third parties (competing STFs) can challenge blocks, by placing a
stake.
The arbitration process is a two-player game played between a claimer
and a challenger, which is governed by an L1 smart contract.
The arbitration process satisfies that a single honest participant can
always win.
If the claimer loses, the block in dispute is removed and the
claimer's stake is lost; if the claimer wins, the challenger loses the
stake.
The winning party receives part of the opponent's stake as a
compensation.

The most prominent Optimistic Rollups based on their market
share~\cite{l2beat} are Arbitrum One~\cite{ArbitrumNitro}, Optimism
mainnet~\cite{optimism} and Base~\cite{base}, while popular ZK-Rollups
include Starknet~\cite{starknet}, zkSync Era~\cite{zksyncera} and
Linea~\cite{linea}.

%
%
In most rollups, a single computing node acts as sequencer, receiving
user transaction requests and posting compressed batches of
transaction requests to L1.
The centralized nature of these sequencers makes these rollups a
\emph{non-decentralized} L2 blockchain, where the sequencer is a
central authority.
%
%
To increase scalability further, other rollups implement sequencers
posting hashes of batches, instead of compressed batches,
dramatically reducing the size of the blockchain interaction.
Working with hashes posted to L1 requires an additional data service,
called \emph{data availability committee} (DAC), to translate hashes
into the corresponding batches.
If either the sequencer or the DAC are centralized\footnote{We refer
  here as centralized solutions even if there is a fixed collection of
  servers picked by the owners of the system.}, the solution is still
a non-decentralized L2 blockchain.

ZK-rollups that rely on DAC are known as \emph{Validiums}.
Current validiums include Immutable X~\cite{immutablex}, Astar
zkEVM~\cite{astar}, and X Layer~\cite{xlayer}.
Optimistic rollups that use DAC are called \emph{Optimiums}, which
include Mantle~\cite{mantle}, Metis~\cite{metis} and
Fraxtal~\cite{fraxtal}.
A list of various Ethereum scaling solutions and their current state
of decentralization can be found in
Appendix~\ref{app:decentralizationL2Eth}.
To simplify notation, in the rest of the paper we use simply L2 to
refer to optimistic rollups that post hashes and have a DAC, that is
Optimiums.
We also use the term \emph{arranger} to refer to the combined service
formed by the sequencer and the DAC.

\vspace{0.5em}
\noindent\textbf{The Problem.}
%
Our goal is to reduce the control power of centralized sequencers.\footnote{However, we do not attempt to solve more complex problems
  such as order-fairness~\cite{kelkar2020orderfairness}.}

\vspace{0.5em}
\noindent\textbf{The Solution.}
%
We describe how to implement a \emph{fully decentralized} arranger.
We rely on algorithms whose trust is in the combined power of a
collection of distributed servers where correctness requires more than
two thirds of the servers run a correct version of the protocol.
%
Our main building block is Setchain~\cite{capretto22setchain}, a
Byzantine resilient distributed implementation of grow-only sets.
Using setchain, users can inject transaction requests through any
server node where they remain unordered until a ``barrier'', named
epoch, is injected.
Setchain is reported in~\cite{capretto22setchain} to be three orders
of magnitude faster than consensus, by exploiting the lack of order
between barriers to use a more efficient distributed protocol to
disseminate transactions requests.
We build in this paper a decentralized arranger using setchain, by
using epochs to determine batches.
Since all correct arranger nodes include a correct setchain node, they
agree on each batch and can locally compute the hash of a batch, and
reverse hashes into batches.
It is well known that Byzantine distributed algorithms do not compose
well in terms of correctness and
scalability~\cite{capretto22setchain}, partly because the different
components need to build defensive mechanisms to guarantee
correctness.
Here we implement an arranger by combining in each distributed node
the role of a sequencer and a DAC member, so a correct server node
includes a correct implementation of a setchain, and a correct
participant in the sequencer and DAC services.
Setchain (and in consequence the decentralized arranger) is a
permissioned protocol, so servers participating are fixed and known to
all other participants.
Protocols to elect and replace arranger nodes are out of the scope of
this paper.

Servers can also post hashes to L1 with the guarantee that all correct
servers know the reverse translation.
Agents, like STFs, can contact any setchain server to obtain the batch
of a given hash.
%

Finally, we introduce \emph{incentives} for servers forming the
arranger to behave properly, for example, including payments for
generating and signing hashes, posting correct batches into L1, and
performing reverse translations.
These servers place stakes when posting batch hashes and are rewarded
for batches that consolidate.
If a batch or hash is proven incorrect all servers involved lose their
stake.
To the best of our knowledge, this is the first paper to introduce
incentives for arrangers (sequencer and DACs) of
decentralized L2.
Our incentives are general and do not depend on the concrete
implementation of the sequencer or DAC (they do not depend on the use
of setchain to implement and arranger presented in this paper).

\vspace{0.5em}
\noindent\textbf{Contributions.}
%
In summary, the contributions of this paper are:
\begin{compactitem}
\item A formal definition of L2 arrangers;
\item A fully decentralized arranger based on setchain and a proof of
  its correctness;
\item Economic incentives including payments and fraud-proof
  mechanisms to detect protocol violations and punishments for
  incorrect servers;
\item An analysis of two adversary models and their impact on the
  correctness of our solution;
\item Implementations of all building blocks to empirically assess
  that they can be used to implement an efficient scalable
  decentralized arranger.
\end{compactitem}

%% file: Preliminaries.tex
\section{Preliminaries}\label{sec:prelim}

In this section, we briefly present our assumptions about the underlying
execution system, Optimistic Rollups and Optimiums, and Setchain.

\subparagraph{Assumptions of the Model of Computation.}
Distributed systems consist of processes (clients and servers) with an
underlying network that they use to communicate, using message
passing.
Processes compute independently, at their speed, and
their internal states remain unknown to other processes.
Message transfer delays are arbitrary but finite and remain
unknown to processes.
Servers communicate among themselves to implement distributed
data-types with certain guarantees and clients can communicate with
servers to exercise such data-types.

Processes can fail and behave arbitrarily.
%
We bound the number of Byzantine failing servers $f$ by a fraction of
the total number of servers.
In Section~\ref{sec:incentives} we introduce \emph{economic incentives
  and deferrals} that motivate servers to be honest.
All processes (including clients) have pairs of public and private
keys, and public keys are known to all processes.
The set of servers is fixed and known to all participants.
We assume reliable channels between non-Byzantine processes, messages
are not lost, duplicated or modified, and are authenticated.
%
%
Messages corrupted or fabricated by Byzantine processes are detected
and discarded by correct processes~\cite{Cristin1996AtomicBroadcast}.
Since each message is authenticated, communication between correct
processes is reliable.
Additionally, we assume that the distributed system is partially
synchronous~\cite{Crain2021RedBelly,Fischer1985Impossibility}.
%
%
Finally, we assume clients create valid elements and servers can
locally check their validity.
In the context of blockchains, this is implemented using public-key
cryptography.

%

  

%
\begin{figure}[t!]
  \centering
  \begin{tabular}{c@{\hspace{1em}}c}
    \includegraphics[scale=0.295]{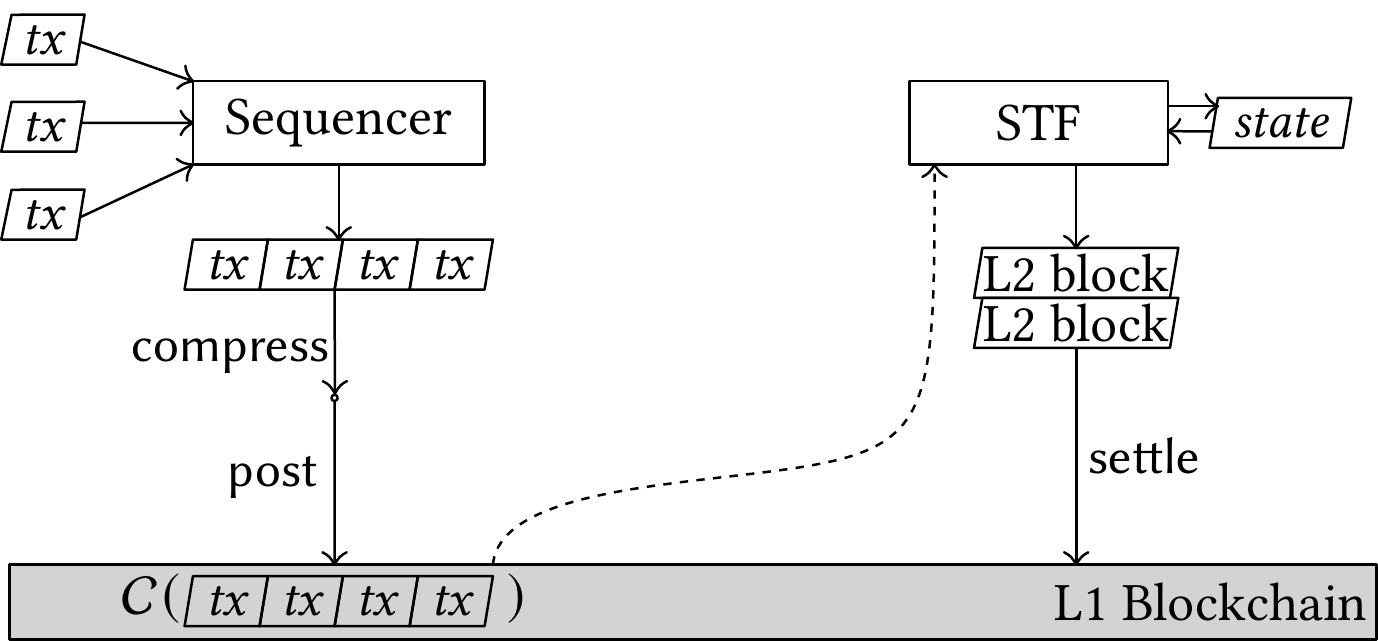}  &
    \includegraphics[scale=0.295]{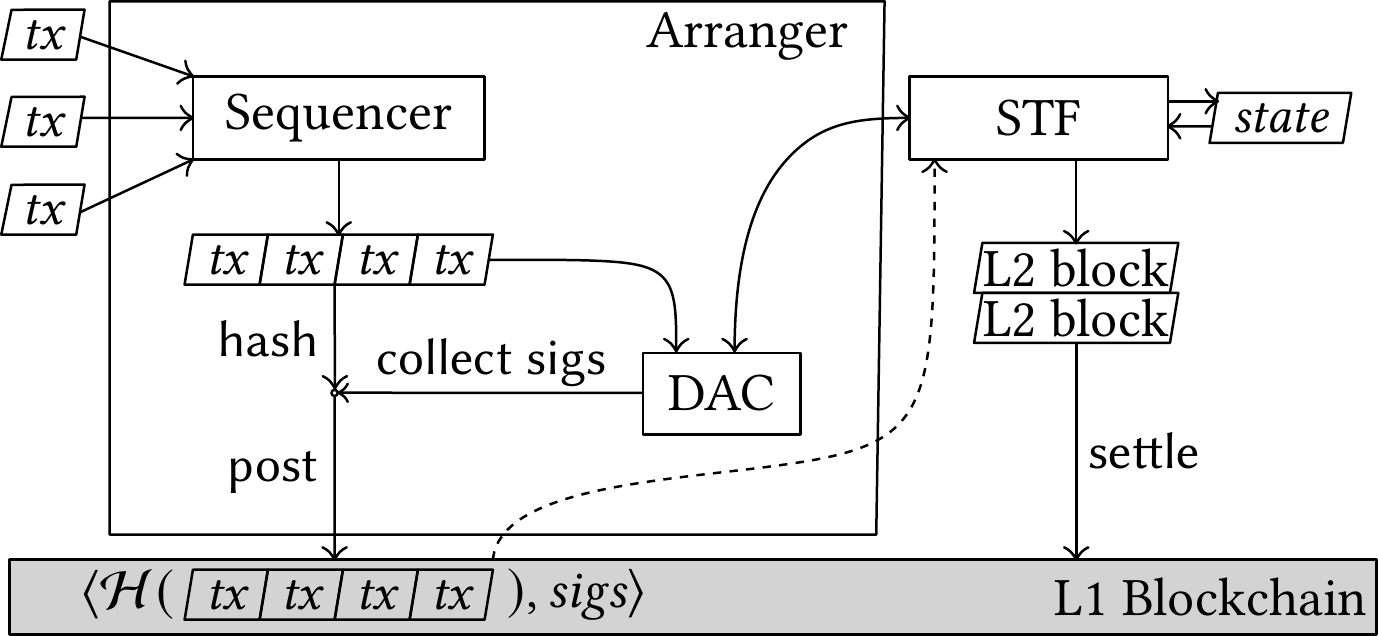}
  \end{tabular}
  \caption{Optimistic Rollups (left) and Optimiums (right).}
  \label{fig:arbitrum-anytrust}
\end{figure}

\subsection{Optimistic Rollups and Optimiums}

The main principle of Optimistic Rollups is to split transaction
sequencing from transaction execution whose effects are delayed to
allow arbitration.
%
Optimistic Rollups are implemented as two components (see
Fig.~\ref{fig:arbitrum-anytrust} (left)):
a \emph{sequencer}, in charge of ordering transactions, and a
\emph{state transition function} (STF), responsible for executing
transactions (this terminology was introduced by
Arbitrum~\cite{Kalodner2018Arbitrum}).
%
%
The sequencer collects transaction requests from users, packs them
into batches, and posts compressed batches---using a reversible
compression algorithm~\cite{Alakuijala18brotli}---as a \emph{single}
invocation into L1 (Ethereum).
Once posted, compressed batches of transactions are immutable
and visible to anyone.

STFs compute the effects of transaction batches which are completely
determined by the sequence of transactions and the previous state of
the L2 blockchain.
%
STFs are independent processes that can propose new L2 states as L2
block \emph{assertions} into L1.
L2 blocks are optimistically assumed to be correct.
STF processes can challenge L2 blocks posted by other STFs during a
challenge period (typically more than one day).
If an STF challenges an assertion, a bisection protocol---which
guarantees that an honest players can win---is played between the STFs
involved.
If the challenger wins, the L2 block is discarded; otherwise it
continues on hold until the challenge period ends.
L2 blocks that survive the challenge period consolidate and the L2
blockchain evolves.

In Optimistic Rollups, the sequencer is a centralized process and does
not offer any guarantee regarding transaction orders or whether
transaction requests can be discarded.
To prevent censorship from the sequencer, some rollups allow users to
post transaction requests directly to L1.

%
Even though batches are compressed in Optimistic Rollups, the L1 gas
associated with posting batches is still significant.
To mitigate this cost, Optimiums introduce a \emph{Data Availability
  Committee} (DAC) storing batches and providing them upon request
(see Fig.~\ref{fig:arbitrum-anytrust} (right)).
In current Optimiums, the sequencer posts a hash of a batch---which is
much smaller than the compressed batch---along with evidence that the
batch is available from at least one honest DAC server.
Some implementations, to guarantee progress, provide a fallback
mechanism, such that the sequencer either collects enough signatures
within a specified time frame or posts the compressed batch into L1.
STFs then check that hashes posted have enough valid signatures and
request the corresponding batch to DAC members.
Recall that in the rest of the paper we will use ``L2'' to refer to the
Optimium extension of Layer 2 Optimistic Rollups that post hashes to
L1 and implements a DAC.

\subsection{Setchain}\label{sec:setchain}
Setchain~\cite{capretto22setchain} is a Byzantine-fault tolerant
distributed data-structure that implements grow-only sets with
synchronization barriers called epochs.
Clients can add elements to any server and servers disseminate new
elements to the other servers.
When barriers are introduced, a set consensus protocol~\cite{Crain2021RedBelly}
is run and all correct servers agree on the content of the new epoch after
which elements can be deterministically ordered.
Setchain efficiently exploits the lack of order
between barriers.
Empirical evaluation suggests that setchain is three orders of
magnitude faster than conventional
blockchains~\cite{capretto22setchain}.


\begin{figure}
  \centering
  \begin{tabular}{|c|c|}\hline
     \includegraphics[scale=0.42]{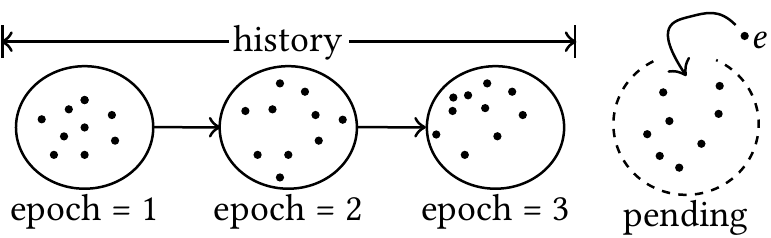} &
     \includegraphics[scale=0.42]{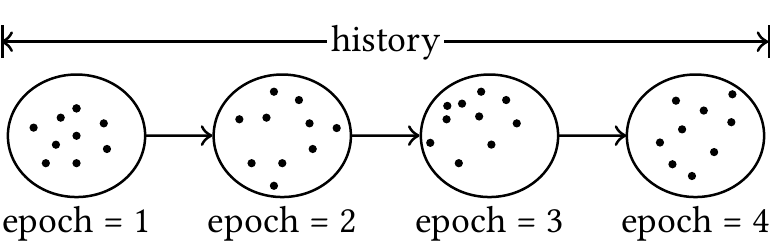}\\ \hline
  \end{tabular}  
  \caption{Applying $\protect\<add>(e)$ after epoch $3$ (left) and
    increasing to epoch $4$ (right). }
  \label{fig:setchain}
\end{figure}

\subparagraph{API.}
Setchain server nodes store locally the following information:
(1) a set of added transaction requests: $\<theset>$;
(2) the current epoch number: $\<epoch>$;
(3) the current mapping from natural numbers to sets of transaction
requests describing each epoch: $\<history>$.
Servers support two operations, always available to any client
process:
operation $\<add>$ adds elements to $\<theset>$;
operation $\<get>$ returns the local values of $\<theset>$,
$\<history>$, and $\<epoch>$, as perceived by the server.\footnote{We
  present this API for simplicity but Setchain supports membership
  queries to get individual epochs.}
We assume epoch changes happen regularly as this can be forced by the
participant servers using a local time-out.

Fig.~\ref{fig:setchain} shows the local view of a Setchain server node
$v$.
Set $v.\<pending>$ contains those elements that $v$ knows but that
have not been assigned an epoch yet:
$v.\<pending> = v.\<theset> \setminus v.\<history>$.
To insert a new element $e$ in the setchain, clients can invoke $v.\<add>(e)$.
Server $v$ checks the validity of $e$, propagates it and adds it to
$v.\<pending>$.
When an epoch increment is triggered, servers propose their elements
in $\<pending>$ and collaboratively decide which of those elements
form the new epoch, see Fig.~\ref{fig:setchain} (right).
Servers also sign the hash of each epoch and disseminate this
signature, for example as a Setchain element.
Clients can contact a single server and authenticate the response by
counting the amount of signatures provided.
%

\subparagraph{Relevant Properties of Setchain.}
%
We list here the properties of setchain relevant for
arrangers in Section~\ref{sec:seq-decentralized-props}.
See~\cite{capretto22setchain} for a more exhaustive list of
properties and proofs.
\begin{compactitem}
\item \textbf{Setchain-GetGlobal:} Adding an element \(e\) to a correct server
  eventually results in $e$ being in $\<theset>$ returned by
  $\<get>$ invocations on every correct server.
\item \textbf{Setchain-GetAfterAdd:} Adding a valid element $e$ to a
  correct server eventually results in $e$ being in $\<history>$
  returned by a $\<get>$ on correct servers.
\item \textbf{Setchain-UniqueEpoch:} Every element is in at most one epoch.
\item \textbf{Setchain-ConsistentGets:} All correct servers agree on
  the content of each epoch.
\item \textbf{Setchain-ValidGet:} Each element $e$ returned by a $\<get>$
  operation in correct servers is a valid element.
\end{compactitem}




%% file: API.tex
\section{Arranger}
\label{sec:api}

In this section, we formally define the notion of L2 arranger (or
\emph{arranger} for short) as the service in charge of receiving and
serializing transaction requests, packing them into batches,
efficiently posting them into L1, and making the data available.

An arranger encompasses both the sequencer and DAC of a L2
solution.
%
%
%
We propose an integrated decentralized implementation that
aims to replace both sequencers and DAC.
Arrangers have two kinds of clients: external users injecting L2 transaction
requests and STFs computing L2 blocks.
STF agents need to determine which transaction belongs to a given
batch to compute L2 blocks and to decide if they should dispute L2
blocks proposed by other STF agents.
We say ``the arranger posts to L1'' to refer to the action of some
arranger server process posting to
L1. 

\subsection{Arranger API}
Arrangers use Merkle Trees~\cite{Merkle88} as hash function \(\hash\)
to hash batches where the leaves are the transactions in the batch.
Merkle trees provide an efficient (logarithmic) membership
authenticated check.
We rely on this method to create an arbitration mechanism to dispute
membership where an honest participant can prove the validity of
hashes and prevent the confirmation of invalid hashes (see
Section~\ref{sec:incentives:challenges}).

Arranger servers order transactions, pack them into batches and create
the hashes of each batch, also assigning consecutive identifiers to
consecutive batches.
Periodically, batch identifiers with their hashes and signatures are
posted to L1 invoking a smart contract called \<logger>.  

Arrangers export two operations to all clients:
operation \<add>(\<tx>) to send transaction request \<tx> to the
arranger so \<tx> is added to a future batch;
and operation \<translate>($\<id>,\<h>$) which requests the batch of
transactions corresponding to hash \(\<h>\) with identifier \(\<id>\).
If  \(\<id>\) does not match any batch, the arranger returns
error \<invalidId>.
If there is a batch \(b\) with identifier \<id>, the arranger returns
\(b\) if it hashes to \(\<h>\) or an error \<invalidHash> otherwise.

When the arranger receives enough transactions (or a timeout is
reached), all correct servers decide a new batch $b$, order
transactions in \(b\), assign an identifier $\<id>$ to \(b\), compute
$h = \hash(b)$, and create a ``\emph{batch tag}'' $(\<id>,h)$.
Each correct server signs the new tag and propagates the signature.
When enough signatures are collected, a combined signature
$\sigma=\<sigs>(\<id>, \<h>)$ is attached to the tag to form a
``\emph{signed batch tag}'' $(\<id>,\<h>,\<sigs>(\<id>,\<h>))$, which
is posted to L1.
The logger contract (in L1) accepts the signed batch tag but
\emph{does not} perform any validation check, optimistically assuming
that the batch is legal.
An arbitration protocol, described in
Section~\ref{sec:incentives:challenges}, guarantees correctness of the
batches (this is similar to L2 block claims in Optimistic Rollups).
Finally, STFs observe the signed batch tags in L1, contact the
arranger invoking $\<translate>$ to obtain the batch of transactions
and compute the next L2 block or check the validity of posted L2
blocks.

Arranger processes can post multiple batch tags with the same
identifier and the logger contract will accept them all.
However, for each identifier, the STF agents only process the first
valid signed batch tag, discarding the rest.

\subsection{Desired Arranger Properties}
\label{sec:api-properties}
We formally define the properties of \emph{correct} arrangers.
These properties include all \emph{basic properties} and some
safety properties of ideal arrangers presented
in~\cite{motepalli2023sok}.

Signed batch tags are \emph{legal} if they satisfy the following
properties:
  \begin{compactitem}
  \item \PrCertified: have at least \(f+1\) arranger server
    signatures~(see Section~\ref{sec:prelim}) guaranteeing
    at least one correct server signed the batch tag.
  \item \PrValidity: Every transaction in the batch is a valid
    transaction added by a client.\footnote{A transaction is valid
      when is properly formed and signed by the originating client.}
  \item \PrIntegrityOne: No transaction appears twice in its batch.
  \end{compactitem}
  
A batch tag is legal if it appears in a legally signed batch tag.
A legal batch tag can be part of two legal signed batch tags when
signed by a different subset of arranger processes.
However, all legal batch tags with the same identifier have the same
batch of transactions because all correct servers agree on the batch
content.
This ensures a deterministic evolution of the blockchain.
Formally, let $(\<id>, \<h>_1, \sigma_1)$ and
$(\<id>, \<h>_2, \sigma_2)$ be two legally signed batch tags posted by
the arranger, such taht $\sigma_1=\<sigs>(\<id>,\<h>_1)$ and
$\sigma_2=\<sigs>(\<id>,\<h>_2)$ are signed by at least $f+1$ nodes, then
\(\<h>_1 = \<h>_2\).

\begin{property}
 \label{prop:uniqueBatch}
 \PrUniqueBatch: all legal batch tags with the same identifier
 correspond to the same batch of transactions.
\end{property}

The following property prevents transaction duplication, censorship
and guarantees data availability.
%

\begin{property} Correct arrangers satisfy:
  \label{prop:basic}
  \begin{compactitem}
  \item \PrTermination: all added valid transactions eventually appear
    in a legal batch tag.
  \item \PrIntegrityTwo: no transactions appears in
    more than one legal batch tag.
 \item \PrAvailability: every legal batch tag posted can be translated into
  its batch.
\end{compactitem}
\end{property}
\PrAvailability is expressed formally as follows. Let $(\<id>, \<h>, \sigma)$ be a
legal batch tag posted by the arranger, s.t. $\<h>=\hash(b)$.
Then the arranger, upon request $\<translate>(\<id>,\<h>)$, returns $b$.
This prevents halting the L2 blockchain by failing to provide
batches of transactions from hashes.

\PrValidity, \PrIntegrityOne, \PrIntegrityTwo, \PrUniqueBatch and
\PrAvailability are safety properties and \PrTermination is a liveness
property, together they define correct arrangers.
%
%
However, the properties listed above do not guarantee any order of
transactions posted by clients.
%


%% file: SeqDecentralized.tex
\section{A Fully Decentralized Arranger Implementation}
\label{sec:seq-decentralized}
Implementing a fully centralized arranger is straightforward.
Similarly, an arranger that consists of a sequencer and a
decentralized DAC is also straightforward from the point of view of
arranger.
The central sequencer handles transaction ordering, batch creation,
hashing, and communicates with each DAC member to share batches and
obtain their signatures (see Appendix~\ref{sec:seqDC} for a more
detailed description).
The real challenge resides in implementing a correct arranger where
both the sequencer and the DAC are decentralized.

\begin{wrapfigure}[5]{l}{0.49\textwidth}
  \vspace{-1.2em}
  \includegraphics[scale=0.32]{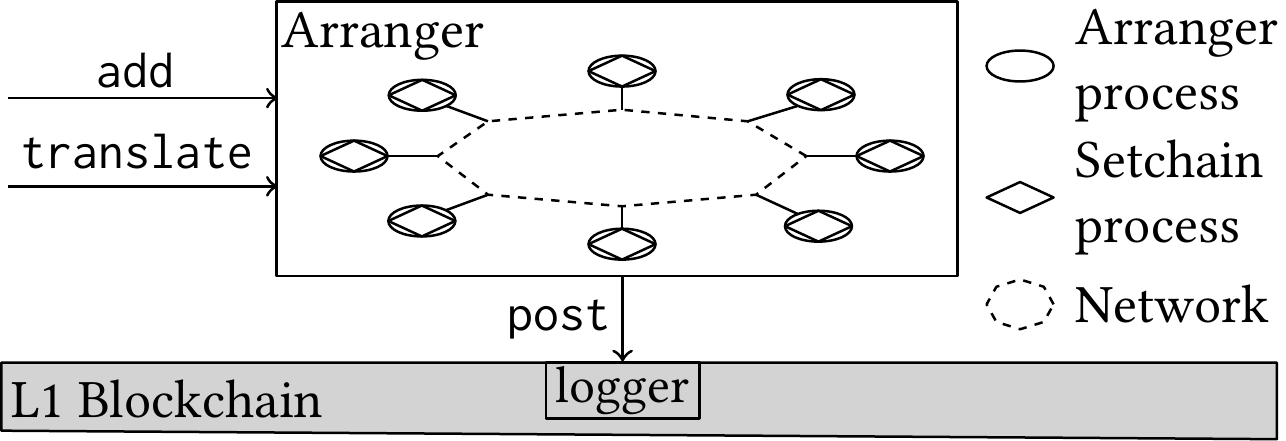}
\end{wrapfigure}

We present now a \emph{fully decentralized} arranger implementation
using Setchain satisfying all correctness properties listed in
Section~\ref{sec:api-properties}.
Both \<add> and \<translate> are implemented
in a completely decentralized fashion.

For simplicity, we assume that our implementation (see the figure above) uses Setchain epochs as batches.
%
%
In our fully decentralized arranger, all servers perform both the
roles of sequencers and DAC members, and take turns to post data to
L1.


\subsection{Implementation}
Algorithm~\ref{alg:arranger-setchain} shows the pseudo-code of correct
servers of decentralized arrangers.
Arranger processes \emph{extend} Setchain
servers~\cite[Alg.~6]{capretto22setchain}~(see
Appendix~\ref{sec:setchain-implementation}).
Here, all servers cryptographically sign each epoch hash and
insert the signature back into the Setchain as an element.
This allows servers to locally provide cryptographic guarantees to
clients about the correctness of an epoch simply by exhibiting $f+1$
signatures.
Servers running Algorithm~\ref{alg:arranger-setchain} maintain a local map
\<hashes> storing inverse translations of hashes and a local set
$\<signatures>$ to store cryptographically signed epochs.
%
\input{algorithms/arranger-setchain}

Distributed arranger processes implement \<add> by adding transaction
requests as elements into the Setchain which, in turn, triggers
communication with other processes maintaining the Setchain.
%
%
When a new epoch is decided by the Setchain internal set consensus,
each process locally receives a \(\<newepoch>\) event with the epoch
number and its content, and $\<tobatch>$ is invoked.
This function separates transactions contained in epochs from elements
serving as ``metadata'' (for example, server signatures of previous
epochs).
The resulting transactions are ordered and packed into a new batch and
\<hashes> is updated.
For simplicity in the explanation, we associate hashing epochs
in~\cite[Alg.~6]{capretto22setchain} with hashing batches, that is,
for each epoch \(e\) with hash \(h_{s}\), as used in Setchain, we
assume that \(h_{s} = \hash(\<tobatch>(e))\).
Hence, when a new epoch is created, each server cryptographically
signs the batch hash and inserts the signature as an element in the
Setchain.
Each server receives an \(\<added>\) event informing the arranger
about cryptographically signed batch hashes and the arranger
sub-process stores the signatures received in the local set
$\<signatures>$.

Arranger servers take turns posting legal signed batch tags into the
logger contract in L1, slicing time and distribute time slices
between them, which is feasible because we assumed the system to be
partially synchronous (and slices can be reasonably long intervals,
like minutes).
Event \<myturn> tells processes when it is their turn to post a
new batch tag to L1.
This is a simple way to guarantee eventual progress.
Even if Byzantine processes do not post legal batch tags in their
turns, the rounds will continue to advance until the turn of a correct
process arrives to post legal batch tags (see
Lemma~\ref{l:termination} below).
Function $\<nextBatch>$ returns the next batch identifier to be posted
to the logger contract which is the minimum identifier for which
there is no legal batch tag in L1.
When it is the turn of arranger server $P$ to post and $P$ has
received enough signatures for the batch, $P$ computes a combined
signature $\sigma$ and posts the legally signed batch tag
$(\<batchId>,h,\sigma)$ into the logger L1 contract.

\subsection{Proof of Correctness}\label{sec:seq-decentralized-props}
Algorithm~\ref{alg:arranger-setchain} implements a correct
arranger.
We inherit the assumptions from Setchain, so for $N$ total servers,
the maximum number of Byzantine processes is $f$ such that $f < N/3$.

We first prove that each legally signed batch tag corresponds with a
Setchain epoch.
Let \((\<batchId>, \hash(b),\sigma))\) be a legally signed batch tag
posted into L1, then $\<tobatch>(e) = b$ in the Setchain where epoch
\(e\) corresponds to epoch number $\<batchId>$.

\begin{lemma}\label{l:batch-epoch}
  Batches of transactions in legal batch tags correspond to
  transactions in Setchain epochs.
\end{lemma}

\begin{proof}
  Legal batch tags require at least one correct server signature and
  correct servers only sign epoch hashes.
\end{proof}

Property~\PrUniqueBatch follows directly from the previous lemma, as
the batch of transactions of each legal batch tag with identifier
\(\<batchId>\) corresponds to epoch \(e\).
Moreover, due to property \textbf{Setchain-UniqueEpoch}, no
transaction in a legal batch appears twice in the same batch or
appears in another legal batch, satisfying properties~\PrIntegrityOne
and \PrIntegrityTwo.

Byzantine servers can re-post existing legally signed batch
tags, but they can not forge fake batches or fabricate legally signed
batch tags because this requires collecting $f+1$ signatures.
Hence, all Byzantine servers can do is to post the same batch tag
several times (perhaps with different collections of $f+1$
signatures).
Since STFs process one legal batch tag per identifier, no transaction
is executed twice and repeated batches are discarded.
In Section~\ref{sec:incentives:challenges} we introduce arbitration
mechanisms penalizing arranger servers for posting illegal batches.

Combining property \textbf{Setchain-ValidGet} (see
Section~\ref{sec:setchain}) with Lemma~\ref{l:batch-epoch}, legal
batch tags only contain valid transactions, so \PrValidity holds.

\begin{lemma}\label{l:termination}
  All Setchain epochs become legal batch tags.
\end{lemma}

\begin{proof}
  Let $\<batchId>$ be the minimum epoch number for which there is no
  legal batch tag posted and let $e$ be $\<batchId>$-th epoch.
  Property \textbf{Setchain-ConsistentGets} ensures all correct
  servers agree on the content of $e$.
  Setchain servers cryptographically sign epoch hashes and insert
  their signature in the Setchain as elements.
  By~\textbf{Setchain-GetGlobal}, all servers eventually receive the
  signatures injected.
  Let \(s\) be the first correct server in charge of posting batch
  tags to the logger after receiving at least $f+1$ signatures of
  epoch $e$.
  There are two cases: either $\<nextBatch>$ returns $\<batchId>$ or a
  larger identifier.
  In the former case, $s$ has enough signatures to generate a legal
  batch tag for $e$ and posts it in the logger contract
  (lines~\ref{alg:arranger-setchain-upon-post-begin}-\ref{alg:arranger-setchain-upon-post-end}).
  In the latter case, a legal batch tag with identifier $\<batchId>$
  has already been posted corresponding to epoch \(e'\)~(see
  Lemma~\ref{l:batch-epoch}).
  By~\textbf{Setchain-ConsistentGets}, we have that $e'=e$.
\end{proof}

Transactions added to at least one correct server eventually appear in
a legal batch tag, which follows from the previous lemma and
property~\textbf{Setchain-GetAfterAdd}.
Therefore, clients contacting one correct server are guaranteed
\PrTermination.
Clients can make sure that they contact at least on good server
contacting $f+1$ servers.
Alternatively, they can contact one server and check that the
answer is correct by checking the cryptographic signatures of the
epochs reported (and contact another server if the check fails).



Finally, we prove that every hash in a legal batch tag posted by the
arranger can be resolved by at least one correct arranger server,
satisfying Property \PrAvailability.

\begin{lemma}~\label{lem:one-correct} For every legal batch tag posted
  by the arranger, at least one correct server $A$ signed it and
  therefore $A$ can resolve its hash.
\end{lemma}

\begin{proof}
  Since legal batch tags have $f+1$ signatures, at least one signature
  comes from a correct server.
  By Lemma~\ref{l:batch-epoch}, all legal batch tags correspond to
  Setchain epochs.
  When Setchain computes an epoch, it triggers an event
  \(\<newepoch>\) signaling the arranger to add the hash to its local
  map $\<hashes>$.
  %
  %
  Upon request, function $\<translate>$ uses $\<hashes>$ to obtain the
  batch that corresponds to a given hash.
  Correct servers never delete information from $\<hashes>$, so for
  every legal batch tag posted at least one of the correct servers
  that signed the batch tag can resolve its hash.
\end{proof}

Lemma~\ref{lem:one-correct} guarantees that when posting legally
signed batch tags to the logger contract, at least one correct server
signed the batch, so the same correct server
is capable of performing the translation.
Eventually, \emph{all} correct servers will know the contents of every
epoch so they will be able to translate as well.

Through this section, we assumed $f < N/3$.
In the next section, (1) we explore incentives that motivate rational
servers to behave correctly, (2) we present a protocol based on
zero-knowledge contingent payments for clients to pay servers for
correctly translating hashes, and (3) we present fraud-proof
mechanisms to detect violations of the protocol.
Then, in Section~\ref{sec:threat-model} we study adversaries that
control the number $f$ Byzantine servers and can potentially bribe the
remaining honest servers and the consequences.


%% file: algorithms/arranger-setchain.tex
\begin{figure}[b!]
  \begin{algorithm}[H]
    \caption{\small Server implementation of decentralized \arranger
      extending Setchain} 
    \label{alg:arranger-setchain}
    \small
    \begin{algorithmic}[1]
        \State \textbf{extends } Setchain
        \State  \textbf{init:} $\<hashes> \leftarrow \emptyset, \<signatures> \leftarrow \emptyset, \<localsetchain>.init$
        \Function{\textnormal{\<add>}}{$t$} \label{alg:arranger-setchain-add-begin}
          \State $\<localsetchain>.\<Add>(t)$ \label{alg:arranger-setchain-add}
          \State \<return>
        \EndFunction \label{alg:arranger-setchain-add-end}
	\Function{\textnormal{\<translate>}}{$\textit{id},h$} \label{alg:arranger-setchain-translate-begin}
           \If{$(\<id>,\_)\notin \<hashes>$}
             \<return> \<invalidId>
           \Else ~\If{$(\<id>,h)\notin \<hashes>$}
                \<return> \<invalidHash>
           \Else
               ~\<return> $\<hashes>(\<id>,h)$
           \EndIf
           \EndIf
          \EndFunction \label{alg:arranger-setchain-translate-end}
          \Upon{$\<newepoch>(\<id>,E)$} \label{alg:arranger-setchain-upon-newepoch}
          \State $b \leftarrow \<tobatch>(E)$
          \State $h \leftarrow \hash(b)$
          \State $\<hashes> \leftarrow \<hashes> \cup \{(\<id>,h), b\}$
          \EndUpon \label{alg:arranger-setchain-upon-newepoch-end}
          \Upon{$\<added>(epoch\_signature(id,h,sig_j(\<id>,h)))$}\label{alg:arranger-setchain-upon-newsig}
          \State $\<signatures> \leftarrow \<signatures> \cup
          \{(id,h,sig_j(id,h))\}$
          \EndUpon \label{alg:arranger-setchain-upon-newsig-end}
          \Upon{$\<myturn>$} \label{alg:arranger-setchain-upon-post-begin}
          \State $id \leftarrow \<nextBatch>()$
          \State \<wait> for $f+1$ signatures
          $(id,h,sig_j(id,h))$ in \<signatures>
          \State \<logger>.\<post>$(id,h,
          \<BLSAggregate>(\bigcup_i
            sig_i(id,h)))$ \label{alg:arranger-setchain-post}
          \EndUpon \label{alg:arranger-setchain-upon-post-end}
      \end{algorithmic}
  \end{algorithm}\vspace{-2em}
\end{figure}


%% file: Incentives.tex
\section{Fraud-proof Mechanisms and Incentives}
\label{sec:incentives}

So far we assumed a bound number of Byzantine servers $f$ without
reasoning why would servers should be honest, in this section, we
incentive servers to be honest.
We provide incentives based on (1) evidence of work with payments and
(2) evidence of malicious behavior with penalties.
We present fraud-proof mechanisms to detect malicious behavior where a
single \emph{honest agent} with enough tokens to participate in the
challenge can uncover malicious behavior.
We perform a simple analysis based on the assumption that rational
agents tend to engage in actions that directly enhance their profit
and refrain from activities that incur in losses.
A more detailed quantitative study including expected utility is left
as future work.

We assume participating processes have public L2 and L1 accounts where
they can receive their rewards and---depending on the
protocol---provide stakes that can be removed as a penalty.

\subsection{Batch Staking and Confirmation}
\label{sec:incentives:stakes}

Malicious arrangers can post illegal batches in L1 as
optimistically no check is performed by the logger contract.
To discourage this behavior, arranger processes stake on the batches
they post.
When an STF posts an L2 block, it places a stake in the L2 block and
implicitly on the batch from which the transactions are extracted to
compute the block.
Processes regain their stake when signed batch tags are
\emph{confirmed}, that is, when the corresponding L2 block is
confirmed, after a confirmation delay as in Optimistic Rollups and
Optimiums.
When a signed batch tag is confirmed, all other batch tags with the
same identifier are discarded and stakes are returned as well.
Before confirmation, signed batch tags can be challenged, and staking
processes that fail to defend their correctness lose their stake.
If a signed batch tag has no stake, it is discarded, along with all L2
blocks executing its transactions.

In Section~\ref{sec:incentives:challenges} below, we present
fraud-proof mechanisms to unilaterally prove that signed batch tags
are illegal and force staking processes to reveal the transactions
within a given batch.
A challenger that plays properly the challenging protocols always wins
against a defender posting an illegal batch tag, even if the
challenger does not know the batch content.

\subsection{Incentives to Generate and Post Batches}
\label{sec:incentives:generate}
The confirmation of batches generates rewards as follows.
For each confirmed batch posted:
\begin{compactitem}
\item Each participating arranger process receives a constant amount
  of $k_1$ L2 tokens.
\item Each process that signs a batch receives $k_2$ L2
  tokens determined by a monotone function $g$ on the number of
  processes signing $\sigma$ (the more signatures, the higher the
  reward for all signers) and of transactions in batch $b$.\footnote{
    Function $g$ is monotone if for all $x,x',y,y'$, if $(x \leq x'$ and
    $y < y')$ or $(x < x'$ and $y\leq y')$ then $g(x,y) < g(x',y')$.}
  This function is fixed and known a priori.
  Processes signing batches receive $k_1+k_2$ L2 tokens for each
  confirmed batch.
\item The process posting the batch in L1 receives an additional
  payment $k_3$ L2 tokens to cover the cost of posting transactions in
  L1.
  This process receives $k_1+k_2+k_3$ L2 tokens.
\end{compactitem}
This mechanism creates incentives to (1) participate in the
decentralized protocol maintaining the arranger, (2) sign batches,
communicate and collect signatures, and (3) include as many
transactions as possible in a batch to be posted in L1.

All payments are charged to the external users posting transactions
and made effective by the STFs when computing the effects of a batch,
including these effects as part of a new L2 block.
Additional L2 tokens can be minted by the L1 contract governing the L2
blockchain.

\subsection{Incentive to Translate Hashes into Batches}
\label{sec:incentives:translate}

%
We present now an additional protocol for servers to translate hashes
and get paid using zero-knowledge contingent
payments~\cite{campanelli2017zkcontingent,fuchsbauer2019WIIsNotEnough,nguyen2020WIIsAlmostEnough}.
When a client \Ctt requests the translation of a hash $h$ contacting
directly server \Stt, if \Stt knows the inverse translation $b$, \Stt
does not respond with $b$.
Instead, \Stt encrypts $b$ using a secret key $k$, such that
$w = \texttt{Enc}_k(b)$ and computes $y$ such that
$\texttt{SHA256}(k) = y$.
Then, \Stt sends $w$ and $y$ to \Ctt along with a zero-knowledge proof
that $w$ is an encryption of $b$ under key $k$ and that
$\texttt{SHA256}(k) = y$.
Client \Ctt verifies the proof and, if it is correct, \Ctt generates a
contingent payment where \Stt is the only beneficiary.
%
%
Server \Stt can only collect the payment by revealing the secret $k$.
When \Stt reveals $k$, \Ctt learns $k$ (so it can decipher $w$ to
obtain $b$) and \Stt receives the payment.
Appendix~\ref{app:htlc} shows the pseudo-code implementing a payment
system based on hashed-timelock contracts.
All communication between \Stt and \Ctt is offchain and one-to-one,
and thus, batch \(b\) is never posted in L1.

Byzantine servers can participate in the first part of the protocol
and then refuse to reveal $k$ causing client \Ctt to incur in
unnecessary costs.
To prevent this, servers must also sign $y$ and send the signature to
\Ctt, so \Ctt has evidence identifying the specific server from which
\Stt obtained $y$.
If the server fails to disclose $k$ after some L1 blocks has passed,
\Ctt can use this evidence to accuse server \Stt, have server \Stt
stake removed and receive some compensation.

\subsection{Fraud-proof Mechanisms to Prevent Illegal
  Batches}\label{sec:incentives:challenges}
We now present fraud-proof mechanisms to deter servers from posting
illegal batch tags.
These mechanisms are similar to the one used in Optimistic Rollups to
arbitrate L2 blocks.
The main difference is that here the arbitration is over concrete
algorithms and not over the execution of arbitrary smart contracts
translated to WASM~\cite{wasm}.

A signed batch tag $t:(\<batchId>,h,\sigma)$ for a batch $b$ is legal
if:
\begin{compactenum}
\item[B1.] The combined signature $\sigma$ contains at least $f+1$
  arranger process signatures.
\item[B2.] Batch $b$ only contains valid transactions.
\item[B3.] No transaction in $b$ is duplicated.
\item[B4.] No transaction in $b$ appears in a previously consolidated
  batch.
\end{compactenum}

\subsubsection{Data Availability Challenge.}
To check a signed batch tag \(t\) is \emph{legal}, processes need to
know the transactions in the batch, but only the batch hash $h$ is
posted as part of $t$.
As a safeguard, we introduce first a \textbf{data availability
  challenge}, which enables any agent \Att to challenge any staking
server \Stt to publish the batch transactions in L1.
The data availability challenge is implemented as a game governed by
an L1 smart contract (see Alg.~\ref{alg:dataChallenge} in
Appendix~\ref{app:challenges-alg} for details).
Both player take turns with a predetermined time for each move.
The game is played over the Merkle tree generated by batch \(b\).

During its turn, \Att can challenge a specific node $n$ in the Merkle
tree.
The node hash, \(h_n\), is available in the L1 contract but \Att does
not know the corresponding batch of transactions.
Initially, only the root hash is known, as it is included in the batch
tag.
When posting a challenge, \Att puts a stake large enough to cover part
of \Stt potential answering expenses.

Server \Stt has two options for responding to a challenge on a node.
If $n$ is a leaf, \Stt posts the transaction \(e\) that \(n\)
represents.
The L1 contract verifies that \(e\) hashes to \(h_n\).
Otherwise, \Stt posts the hashes, \(h_l\) and \(h_r\), of the
challenged node left and right children, along with their
corresponding compressed batches \(b_l\) and \(b_r\).
In this case, the L1 contract simply verifies that \(h_n\) is obtained
by hashing the concatenation of \(h_l\) and \(h_r\), ensuring that
both hashes effectively belong to the Merkle tree.
The decompression of the both batches is done offchain by \Att.
If decompressing \(b_l\) results in a batch hashing to \(h_l\), and
\(b_r\) results in a batch hashing to \(h_r\), \Att just learned all
transactions hashing to \(h_{n}\).
Conversely, if a compressed batch does not match its hash, \Att can
challenge the corresponding node.

If \Stt does not respond in time or the provided information fails the
L1 contract verification, \Att wins the game and \Stt loses it
stake.
Otherwise, \Att will ultimately learn all transaction in \(b\).

\begin{repprop}{DAC} \label{prop:data-availability}
  For each unconfirmed signed batch tag posted in L1, any honest agent
  can force to either learn its transaction or get the batch tag
  removed.
\end{repprop}

%
%
%

In the worst case, every node in the Merkle tree is challenged
once.\footnote{Information is public in L1, so it makes no sense to
  challenge the same node twice.}
That is, the maximum number of turns in this game is linear in the
size of batch \(b\).

A crucial difference between this mechanism and the offchain protocol
described in Section~\ref{sec:incentives:translate} is that this
mechanism provides evidence if a staker does not translate a
hash leading to the removal of its stake as a punishment.
%
%
However, this mechanism requires posting the compressed batch onchain.
Posting onchain is more expensive for clients and there is no rewarding for
servers.
Therefore, under normal circumstances, only the offchain protocol is
used to translate hashes (see Section~\ref{sec:incentives:costs}).

\subsubsection{Fraud-proof Mechanisms to Challenge Batch Legality}
Once \Att obtains the data, it can verify locally the legality of the
batch tag.
The mechanisms are governed by L1 smart contracts that arbitrate the
challenges~(see Appendix~\ref{app:challenges-alg} for their
pseudo-codes).
%
If a tag is illegal, depending on which condition B1, B2, B3 or B4 is
violated, agent \Att can play one of the following challenging games:
%
\begin{compactitem}
\item \textbf{Signature Challenge} (violation of B1).
  We use one L1 contract to check whether $t$ has enough signatures.
  Another L1 contract computes the aggregated public key of the
  claimed signers and checks that the signature provided in $t$ is
  correct.\footnote{Alternatively, a fraud-proof mechanism, arbitrated
    by an L1 contract, can be implemented to arbitrate that the
    algorithm that checks multi-signatures accepts the signature
    provided.}
  If \Att invokes these contracts and it is revealed that the
  signature is incorrect, then the signed batch tag \(t\) is
  discarded, all staking processes lose their stake and \Att gains
  part of the stakes placed on \(t\).
\item \textbf{Validity Challenge} (violation of B2).
  Agent \Att knows there is an invalid transaction $e$ in batch \(t\).
  In this case, \Att challenges a staking server \Stt to play a game
  in L1 similar to data availability challenge.
  %
  %
  This time the game is played bottom-up over the path from the leaf
  where \Att claims that \(e\) is.
  In its first turn, \Att presents \(e\), and its hash \(h_e\), its
  position in the batch, and the hash of the node in the middle of
  the path from the leaf to the root, bisecting the path.
  An L1 smart contract (called the Validity Challenge contract) first
  verifies the validity of transaction \(e\) and checks that its hash
  is \(h_e\).
  This contract is also used to arbitrate a bisection game over the
  path.
  In the subsequent turns, players alternately perform the following
  operation: \Stt selects one subpath and challenges it.
  If the subpath is longer than two, then in its next turn \Att must
  provide the hash of the middle node in that subpath.
  Otherwise, the bisection is over and \Att must provide evidence that
  the subpath is correct.
  For a subpath that has only two nodes with hashes \(h_p\) and
  \(h_c\), where the node with hash \(h_p\) is the parent of the node
  with hash \(h_c\), \Att must post a hash \(h\).
  The Validity Challenge contract verifies that hashing the
  concatenation of \(h\) with \(h_c\) results in \(h_p\).
  If the hash provided is correct \Att wins the game.
  %
  %
  If \Att wins the game then \Stt lose its stake and \Att receives
  part of it.
  Otherwise, \Stt keeps its stake.
  It is important to note that the L1 smart contract verifies hashes
  only during \Att's first turn and when a subpath reaches two nodes.
  Since the path length is halved at each turn, the maximum number
  of turn in the game is logarithmic in the path length which itself
  is logarithmic in the batch size.

\item \textbf{Integrity Challenge 1} (violation of B3):
  Agent \Att knows a transaction \(e\) that appears twice in \(t\).
  In this case, \Att challenges a staking server \Stt to play a game
  in L1 (controlled by the Integrity Challenge 1 contract) similar to
  the validity challenge, but \Att provides two paths in \(t\) where
  \(e\) is a leaf.
  The smart contract does not perform any check in this step.
  For each path \Att presents the batch tag, the hash of \(e\), the
  position of \(e\) in the corresponding batch, and the hash of the
  middle node on the path from \(e\) to the root.
  After that, the game continues exactly as in the validity challenge.
\item \textbf{Integrity Challenge 2} (violation of B4):
  Agent \Att knows a transaction \(e\) in \(t\) that appears in a
  previously confirmed batch tag \(t'\).
  In this case, \Att challenges a staking server \Stt to play a game
  in L1 similar to integrity challenge 1, but in its first turn \Att
  provides not only paths in \(t\) where \(e\) is a leaf, but also
  \(t'\) identifier and the path in \(t'\) where \(e\) is a leaf.
  After that, the game continues exactly as in the integrity challenge
  1.
\end{compactitem}

%

%

In all previous challenges an honest agent \Att that knows the
transactions associated with an illegal batch tag can win the
corresponding challenge resulting in the batch tag being
discarded.

\begin{repprop}{correctness}\label{prop:correctness}
  Let \(t\) be an illegal batch tag that is not yet confirmed in L1,
  and \Att be an honest agent that knows all transactions in \(t\) and
  all transactions in already confirmed batches.
  Agent \Att can prevent the confirmation of \(t\).
\end{repprop}

By combining Propositions~\ref{prop:data-availability}
and~\ref{prop:correctness} illegal batch tags can be removed by a
single honest agent.
%

\begin{corollary}\label{cor:confirm-correct}
  Any honest agent can ensure that all confirmed batches are legal.
\end{corollary}

An honest agent that knows the transactions in a legal batch tag can
guarantee its confirmation if no other valid batch tag with the same
identifier was posted in L1.
This confirmation process involves computing the transactions
effects, posting them to L1, and successfully responding to any
challenge.

\begin{repprop}{correctness2}\label{prop:correctness2}
  Let \(t\) be a legal batch tag still unconfirmed in L1 for which
  there is no previous legal batch tag posted in L1 with the same
  identifier.
  An honest agent, \Att, that knows all transactions in \(t\) and all
  transactions in already confirmed batches, can guarantee its
  confirmation.
\end{repprop}

%
\begin{proof}(sketch)
Since \Att knows all transactions in
\(t\), \Att can compute their effect and post it as an L2 block.
If the L2 block is challenged, \Att can win the challenge by correctly
playing the fraud proof mechanism of the L2 blockchain.
We need to prove that \Att can keep its stake in \(t\) until the
corresponding L2 block is confirmed, that is, that \Att can win each
challenge.
For the \textbf{data challenge} \Att computes the Merkle tree \(mt\)
associated with the transactions in \(t\) and uses $mt$ to answer all
challenges correctly.
Since \(t\) is correctly signed by at least \(f+1\) arranger processes
the L1 smart contract directly passes the \textbf{signature challenge}.
The remaining challenges are the \textbf{validity} and
\textbf{integrity} challenges.
In those cases, the challenger tries to prove that an element that
does not belong to a given Merkle tree \(mt\) is actually in \(mt\).
To win the challenge, \Att maintains the invariant that after its turn
in the challenged subpath the top node belongs to \(mt\) while the
bottom node does not.
   When the challenged subpath has two nodes with hashes \(h_p\) and
   \(h_c\), where the node with hash \(h_p\) is the parent of the node
   with hash \(h_c\), the challenger needs to post a hash \(h\) such
   that the concatenation of \(h\) with \(h_c\) hashes to \(h_p\).
   However, the challenger cannot provide such an \(h\) as, by the
   invariant, \(h_p\) belongs to \(mt\) but \(h_c\) does not.  
 \end{proof}   
 
Combining the previous proposition with Properties \PrUniqueBatch and
\PrTermination of correct arrangers, the following holds:

\begin{corollary}\label{cor:confirm-correct2}
  All valid transactions added to correct arrangers are eventually
  executed.
\end{corollary}


Propositions~\ref{prop:data-availability}, \ref{prop:correctness}
and~\ref{prop:correctness2}, and Corollary~\ref{cor:confirm-correct}
hold even when the number of Byzantine servers in the arranger is not
bounded by $f$.
In other words, it does not matter the amount of Byzantine nodes, \textbf{an
honest agent can guarantee the properties mentioned above}.



%
Current rollups do not provide any mechanism to guarantee that batches
posted by their arranger are correct, because their arranger is a
centralized sequencer which is \emph{assumed} to be correct.
The centralized sequencer in Optimistic Rollups can post data that
does not correspond to a compressed batch of valid transaction
requests.
Similarly, the sequencer of Optimiums can post a hash that is
not signed by enough members of the DAC or whose inverse does not
represent a batch of new valid transactions. 
The implications of these actions are significant, as malicious STFs
could post illegal new L2 states that cannot be disputed by other STF
processes as they do not know the transactions to play the challenging
game.
Hence, a malicious sequencer in current rollups can block the L2
blockchain or, even worse, create indisputable incorrect blocks.
This is not possible in our solution.

\subsubsection{Challenge Mechanisms Cost Analysis}
\label{sec:incentives:cost-analysis}
\label{sec:incentives:costs}
We perform a simple analysis between costs and rewards for each challenge
mechanism and the translation protocol to determine the minimum budget required by
honest agents.
We start by defining the following variables:
\begin{compactitem}
\item variable $s$ is the minimum amount of tokens needed to stake in a batch tag.
\item variable $\CC{translate}$ is the client cost for translating a hash using the
  offchain protocol.
\item variable $\SC{translate}$ is the server cost for translating a hash using the
  offchain protocol.
\item variable $\SR{translate}$ is the server reward for translating a hash using the
  offchain protocol.
\item variable $\CC{x}$ is the client cost for playing challenge of type $x$.
\item variable $\CR{x}$ is the client reward for winning the challenge of type $x$.
\item variable $\SC{x}$ is the staker cost for playing the challenge of type $x$.
\end{compactitem}
Where $x$ can be \KWD{data}, \KWD{signature}, \KWD{validity} or
\KWD{integrity}.

Stakers do not receive any reward for winning challenges, as their
motivation for playing is to not lose their stake and eventually
collect their reward from consolidated blocks.
However, they do receive a reward for translating a hash using the
offchain protocol as there is no stake at risk.

The following relations must hold:
\begin{compactitem}
\item The cost for clients to play any challenge $x$ is significantly
  less than the reward for winning it, $\CC{x} \ll \CR{x}$.
  The cost for servers to translate hashes using the offchain
  protocol is significantly less than its reward
  $\SC{translate} \ll \SR{x}$.
  This motivates rational agents to participate to gain a
  profit with a large margin.
\item For any challenge $x$, clients cover the costs of the
  challenged staker, $\CC{x} > \SC{x}$.
  In particular, no staker plays the signature challenge, therefore
  $\SC{signature} = 0$.
\item For the data, validity and integrity challenges, the client
  reward for winning the challenge is less than the challenged agent
  stake, $\CR{x} < s$.
  For the signature challenge, the client reward is bounded by the sum
  of all stakes in the batch, $\CR{signature} < \sum s$.
  For the offchain translation protocol, clients cover the
  costs of servers reward, $\CC{translate} > \SR{translate}$
  This ensures rewards can be covered by the stake taken from the
  losing player.
\item The client cost for using the offchain translation protocol is
  significant less than playing the data challenge, $\CC{translate} \ll
  \CC{data}$.
  This motivates clients to use the offchain protocol when translating
  hashes.
\end{compactitem}

We now determine the minimum budget a single honest agent needs to
guarantee that an illegal batch tag does not get confirmed.
\begin{repprop}{cost}\label{prop:cost}
  An agent \Att with
  \(max(\CC{signature}, \CC{data} +
  max(\CC{validity},\CC{integrity}))\) tokens knowing the
  transactions in all confirmed batches can discard
  illegal signed batch tags. 
\end{repprop}

Finally, rational agents are likely to use the offchain protocol to
translate hashes, because this protocol is cheaper for clients and
more rewarding for servers than data challenging.

\begin{proposition}
  Let \Att be a rational agent seeking to translate a hash \(h\), and
  \Stt be a rational arranger process that knows the batch hashing to
  $h$.
%
  The translation of \(h\) will occur through the offchain protocol.
\end{proposition}



The proof of all the propositions can be found in
Appendix~\ref{app:ProofIncentives}.


%% file: ThreatModel.tex
\section{Threat Models}
\label{sec:threat-model}

%
In this section, we describe two different adversaries controlling
more than $f$ arranger servers, violating the correctness assumptions
in Section~\ref{sec:seq-decentralized}, and study the impact on our L2
blockchain.

\subparagraph{Adversary \#1: Limited Power.}
%
Lets consider an adversary controlling between \(f+1\) and \(N-f-1\) arranger
processes without compromising the Setchain.
In other words, there are at least \(f+1\) correct arranger processes which
agree on the Setchain epochs.

If adversary arranger processes post illegal batch tags in L1, an honest agent
can detect these tags and prevent their consolidation (see
Corollary~\ref{cor:confirm-correct}).
If they post a legal batch tag not following the batch
reached by Setchain consensus, a correct arranger process can post
an incompatible legal batch tag with the same identifier violating the
Property~\PrUniqueBatch.
This collective arranger behavior is incorrect and detected upon
inspection of the L1 blockchain, i.e. two incompatible signed batch
tags are posted to L1.
In this case, we need a mechanism to replace arranger processes, but
it is out of the scope of this paper.

\subparagraph{Adversary \#2: Fully Powerful.}
We consider now a more powerful adversary controlling between
\(f+1\) and \(N\) arranger proccesses including compromising the
Setchain.
In this case, even if there are correct arranger processes, they may
not agree on Setchain epochs.

Even if the adversary controls all processes, an honest agent can
prevent the execution of invalid transactions or the execution of a
transaction twice (corollary~\ref{cor:confirm-correct}).

However, the adversary decides which transactions appear in legal
batch tags and detecting that the arranger is controlled by an
adversary is impossible.
This gives the adversary the power to censor transactions, but
to provide censorship resistance, users can post their transactions
directly in L1 although at a higher cost.


%% file: EmpiricalEvaluation.tex
\section{Empirical Evaluation}
\label{sec:empirical}

We provide now an empirical evaluation to asses whether the
decentralized arranger from Section~\ref{sec:seq-decentralized} can
scale to handle the current demand of L2.

%
Setchain has been reported to be three orders of magnitude faster than
consensus~\cite{capretto22setchain}.
While Ethereum offers a throughput of approximately \EthTPS{}
Transactions Per Second~(TPS), Setchain can handle $12,000$~TPS.
In comparison,
the cumulative
throughput of all Ethereum L2 systems amounts to \LdosTPS{} TPS.\footnote{Data obtained on May 2024 from l2beat
  \url{https://l2beat.com/scaling/activity}.}
Although Ethereum is operating close to its maximum capacity, the L2
throughput mentioned above represents their \emph{current demand}.\footnote{The available
  documentation of the L2 systems does not present
  theoretical or empirical limitations.}
%
%
Our approach extends Setchain, and thus, we focus on assessing the
efficiency of each new component forming the decentralized arranger.
We aim to evaluate empirically that these components impose a
negligible computational overhead on top of the intrinsic
set-consensus algorithm employed in Setchains.
The empirical evaluation reported in~\cite{capretto22setchain} does
not involve compressing, hashing or signing setchain epochs, nor
verification or aggregation of signatures.

\subsection{Hypothesis and Experiments}
All experiments were carried out on a Linux machine with $256$~GB of
RAM and $72$~virtual 3GHz-cores (Xeon Gold~6154) running
Ubuntu~20.04.
We employed $40,000$ existing transactions from Arbitrum One history
as our data set.

%
\newcommand{\HSize}{H.Size}
\newcommand{\HHash}{H.Hash}
\newcommand{\HCompress}{H.Compress}
\newcommand{\HSign}{H.Sign}
\newcommand{\HAgg}{H.Agg}
\newcommand{\HVer}{H.Ver}
\newcommand{\HTrans}{H.Trans}
%
Our hypothesis is that setchain performance is not affected by:
\begin{compactitem}
\item\textbf{(\HHash)} hashing batches of transactions,
\item \textbf{(\HCompress)} compressing batches of transactions,
\item \textbf{(\HSign)} signing hashes of batches,
\item\textbf{(\HAgg)} aggregating signatures,
\item\textbf{(\HVer)} verifying signatures,
\item\textbf{(\HTrans)} translating hashes into batches.
  %
\end{compactitem}
Additionally, signed batch tags size is significantly smaller than
their corresponding compressed batch (hypothesis \textbf{\HSize}).

\subparagraph{Hypothesis \HSize.}
We measured the size of both compressed batches and batch tags ranging
from $400$ to $4,400$ transactions per batch.\footnote{As of August
  2023, Arbitrum Nitro batches typically contain around $800$
  transactions, while in Arbitrum AnyTrust, batches contain
  approximately $4,000$ transactions.}
For each batch size, we shuffled transactions from our data set,
split them into 10 batches of equal size, and assigned them
unique identifiers.
%
We compressed each batch using the Brotli compression
algorithm~\cite{Alakuijala18brotli}, currently used in Arbitrum
Nitro~\cite{ArbitrumNitro}.
Additionally, for each batch, we computed its Merkle
Tree~\cite{Merkle88} with their transactions as leaves.
To create batch tags, we signed each batch Merkle Tree root along with
an identifier using a BLS Signature scheme~\cite{Boneh2001Short}
implemented in~\cite{ArbitrumNitroGithub}.
Our experiments show there is a ratio of \emph{more than two orders}
of magnitude between compressed batches and batch tags with the same
number of transactions.
Fig.~\ref{fig:size}(a) shows the average size of compressed batches of
400 transactions is approximately 80,000~Bytes, whereas the
average size for 4,400 transactions is around 780,000~Bytes.
In contrast, the average size of signed batch tags remained around
$480$~Bytes independently of the number of transactions within the
batch.
Hypothesis \textbf{\HSize{}} holds, hashes are
significantly smaller than compressed batches.


\begin{figure}[t!]
  \centering
  \begin{tabular}{cc}
  \includegraphics[width=0.475\textwidth]{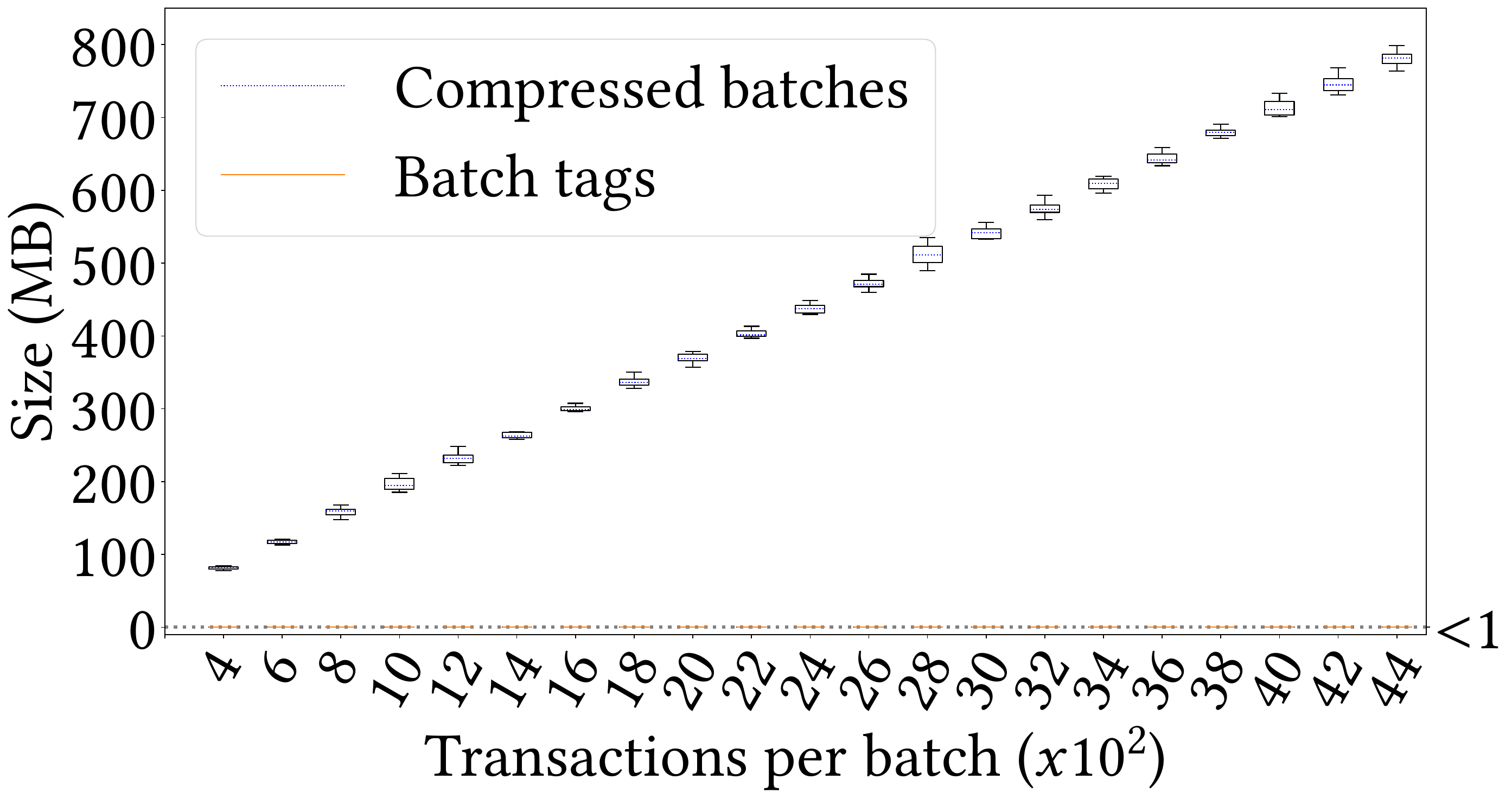} &
   \includegraphics[width=0.475\textwidth]{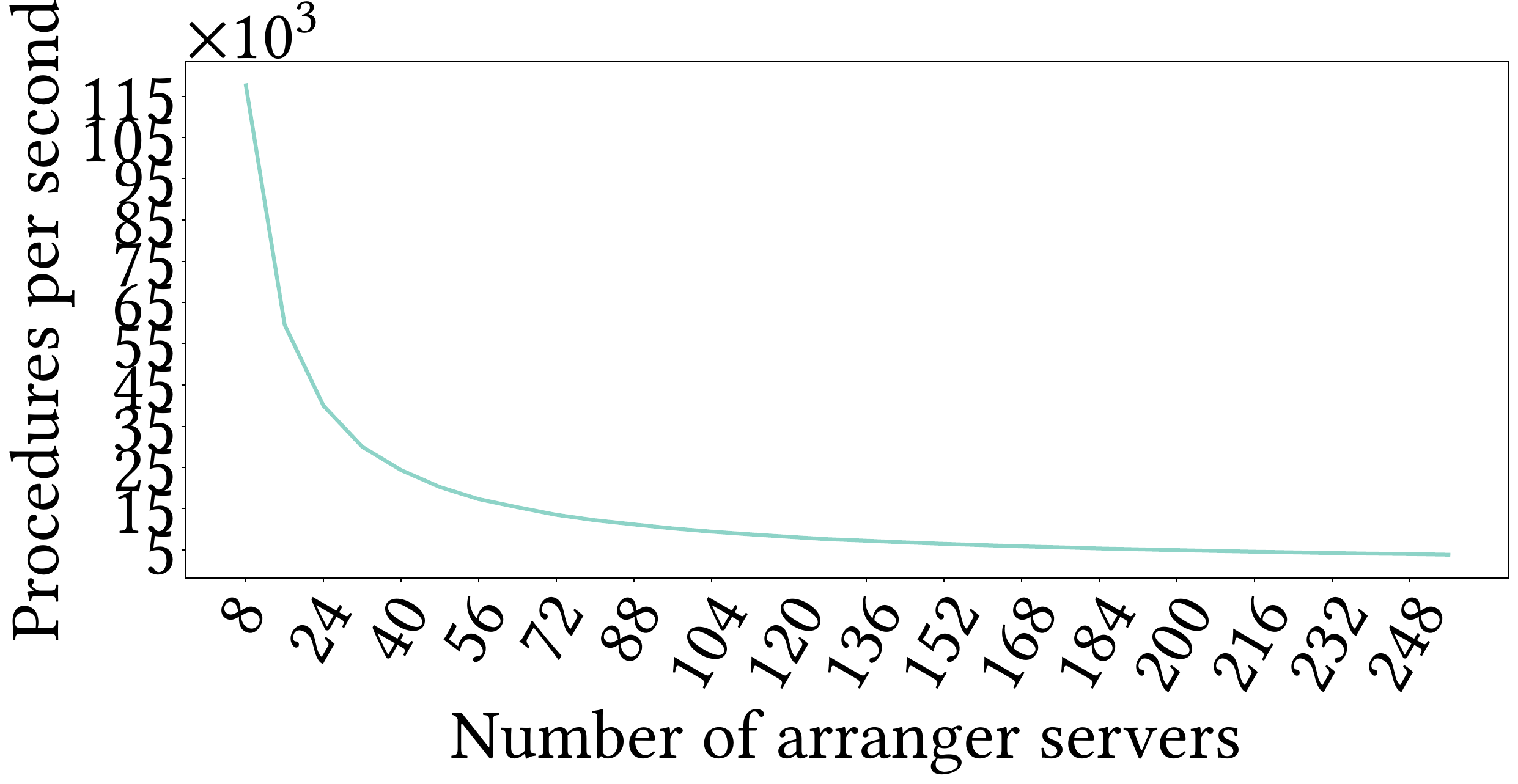} \\
    (a)  & (b) \\ 
  \end{tabular}  
  \caption{(a) Average size of compressed batches and batch tags for varying
  numbers of transactions per batch.
    (b) Signature aggregations per second, in terms of the number
    of signature aggregated per procedure.
  }
  \label{fig:size}
  \label{fig:aggregation}
\end{figure}


For the rest of the experiments, we loaded all the required
information into local memory and applied each procedure for a
duration of $1$~second, cycling over the input as much as necessary.
We ran each experiment $10$ times and report the average performance.
Fig.~\ref{fig:tps-hash-compress-translate}(a), Fig.~\ref{fig:aggregation}(b)
and Fig.~\ref{fig:ver-parallel}(b) summarize the results.

\subparagraph{Hypotheses \HHash{} and \HCompress{}.}
To test compressing and hashing procedures, we use the full data set
of $40,000$~transactions.
We employed Brotli~\cite{Alakuijala18brotli} as compression algorithm
and the root of Merkle trees of batches as hash function.
We again ranged from $400$ to $4,400$ transactions per batch.

Fig.~\ref{fig:tps-hash-compress-translate} shows hashing can
scale from $260,000$ to $275,000$~TPS, depending on the batch size,
while the compression procedure achieves a throughput ranging from
$75,000$~TPS for batches of $400$ transactions to $50,000$~TPS for
batches of $4,400$~transactions.
Both hashing and compression procedures significantly outperform the
Setchain throughput of $12,000$~TPS, and thus, these procedures do not
impact the overall throughput.
Therefore, hypotheses \textbf{\HHash{}} and hypothesis
\textbf{H.Compress} hold.

\subparagraph{Hypothesis \HSign{}.}
To evaluate the signing procedure performance, we used a file
containing $50$ pairs batch hash-identifier.
We assess the performance of signing each pair using the provided
secret key with the BLS signature scheme implemented
in~\cite{ArbitrumNitroGithub}.

Our results show our system can sign approximately $2300$
hash-identifier pairs per second.
Considering a throughput of $12,000$~TPS and batches of at least
$400$ transactions, no more than $30$ batches per second will be
generated.
Thus, each arranger process can sign more than $75$ times the number
of batches required, empirically supporting
hypothesis~{\textbf{\HSign}}.

\begin{figure}[t!]
  \centering
  \begin{tabular}{cc}
   \includegraphics[width=0.475\textwidth]{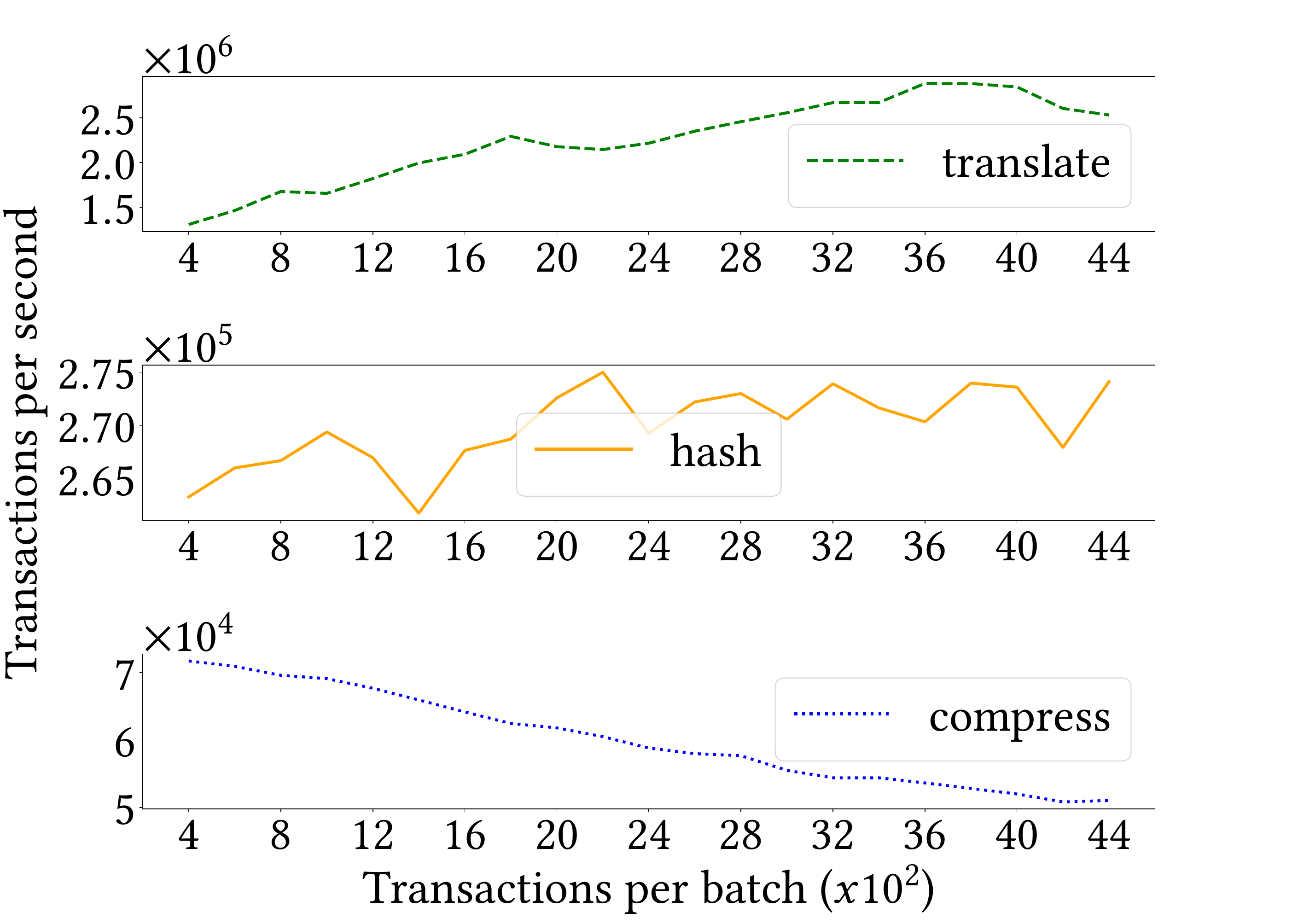} &
    \includegraphics[width=0.475\textwidth]{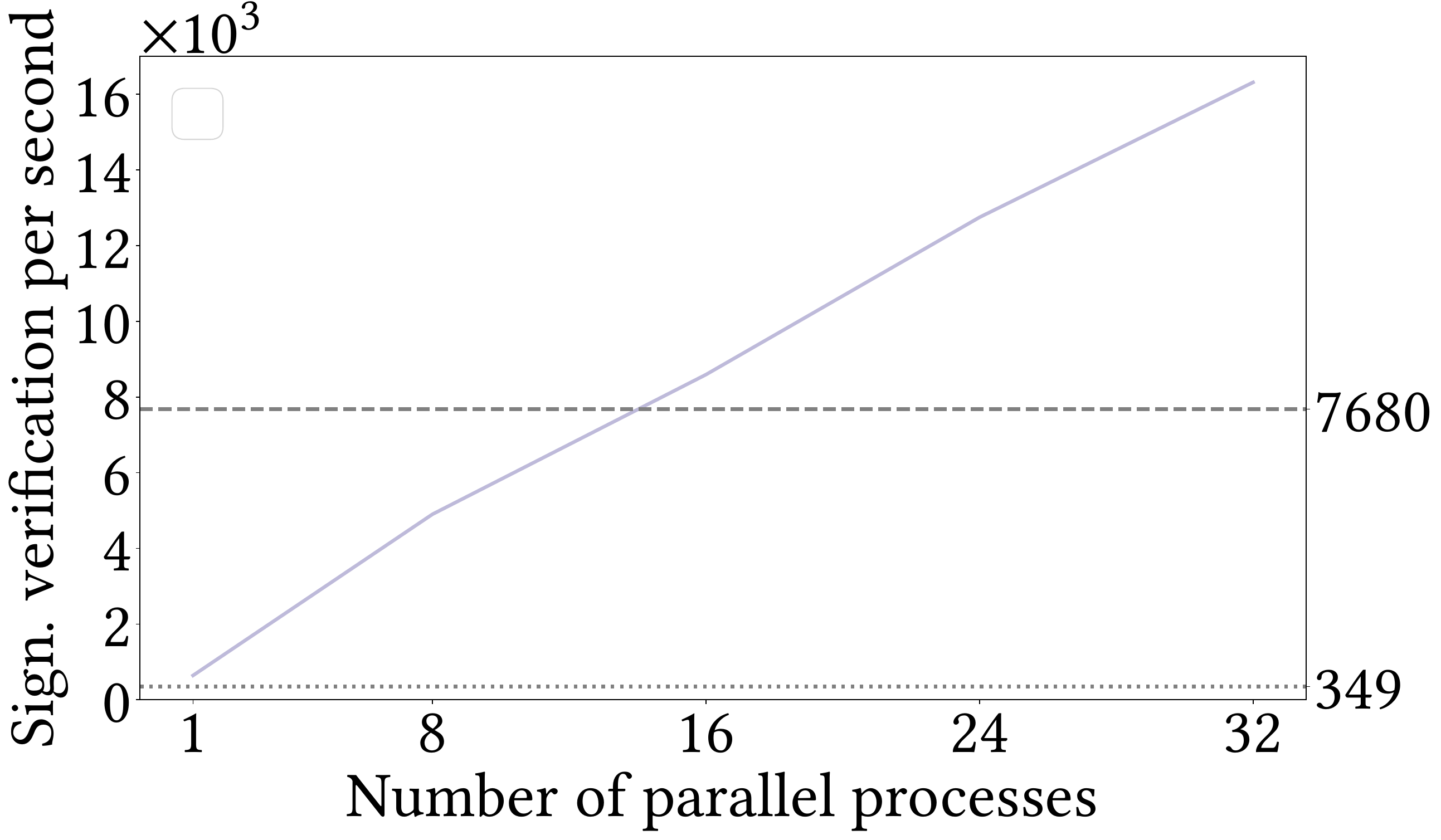} \\
    (a)  & (b) \\ 
  \end{tabular}  
  \caption{(a) Throughput of compressing, hashing and translating
    procedures for varying number of transactions per batch.
    (b) Signature verification per second.
    The dotted line
    represents the maximal number of signatures that each arranger
    process must verify (for $12,000$~TPS, the batch size of
    $4\,400$ and $128$ process). The dashed line represents the maximal
    number of signatures that each arranger must verify (for
    12,000~TPS, the batch size of $400$ and $256$ process).
  }
  \label{fig:tps-hash-compress-translate}
  \label{fig:ver-parallel}
\end{figure}

\subparagraph{Hypothesis \HAgg{}.}
The performance of the signature aggregation procedure depends on the
number of signatures aggregated is limited by the number $N$ of
arranger servers.
%
%
%
We studied the performance of signature aggregation ranging from \(8\)
to \(256\) servers\footnote{As reference, currently,
  the DAC of Arbitrum AnyTrust comprises $7$ members
  \url{https://docs.arbitrum.foundation/state-of-progressive-decentralization\#data-availability-committee-members}.}.
Our data set consists of a directory with $N$ files.
Each file contains $50$ signed hash-identifier pairs signed by
the same signer.
The list of pairs hash-identifier is the same across all files.
We measured the performance of aggregating $N$ signatures of the same
hash-identifier pair using BLS~\cite{ArbitrumNitroGithub}.

Fig.~\ref{fig:aggregation}(b) shows an aggregation ratio ranging from
$5,000$~procedures per second, when aggregating $256$~signatures in
each invocation, to $115,000$~per second when aggregating
$8$~signatures.
Assuming at most $30$ batches per second are generated, the
aggregation procedure can easily handle its workload, even for
$N = 256$ arranger processes, empirically validating hypothesis
\textbf{\HAgg{}}.

\subparagraph{Hypothesis \HVer{}.}
To assess the performance of the signature verification procedure, we
used a file containing $50$~batch tags and the public key of the
signer of all hashes.
We analyzed the performance of verifying that each signature was
generated by signing the corresponding hash-identifier pair using the
corresponding private keys.
We employed again BLS~\cite{ArbitrumNitroGithub}.
Fig.~\ref{fig:ver-parallel}(b) shows that signature verification can be
performed at a rate of approximately $600$~signatures per second.
This rate is sufficient to verify all signatures generated by $N=128$
processes, considering a throughput of $12,000$~TPS and batches
containing $4,400$ transactions (dotted line in
Fig.~\ref{fig:ver-parallel}(b)).
Servers must verify the highest number of signatures when the
arranger is maintained by $256$ processes and batches contain only
$400$ transactions, which amounts to verifying $7,680$ signatures per
second, surpassing the throughput of the signature verification
procedure (dashed line in Fig.~\ref{fig:ver-parallel}(b)).
However, verifying signatures can be done in parallel as the task is
completely independent for each signature.\footnote{All other local
  procedures described can use parallelism easily as well.}

Our empirical evaluation suggests that using $16$~parallel processes
to verify signatures achieves a rate exceeding $8,500$~signature
verified per second, sufficiently for $N=256$ processes and batches
containing $400$ transactions at $12,000$~TPS.
Consequently, the verifying procedure does not influence the
throughput, and hypothesis \textbf{\HVer{}} holds.

\subparagraph{Hypothesis \HTrans{}.}
Finally, we implemented a \emph{translation server} that takes two
input files and generates a dictionary mapping hashes to their
compressed data.
One file contains hashed batches together with identifiers while the
other contains compressed batches with identifiers.
The translation server opens a local TCP connection port to receive
hash translation requests receiving an identifier and returning a
compressed batch according to the dictionary.

Our experiments involve making sequential requests to a single
translation server asking to translate input hashes for a duration of
$1$ second.
All batches contained the same number of transactions, ranging from
$400$ to $4,400$ transactions.

Fig.~\ref{fig:tps-hash-compress-translate}(a) shows that the translation
for batches of $400$ transactions can be performed at a rate of about
$1,300,000$~TPS.
The translation throughput increases as the number of transactions per
batch increases, peaking at nearly $2,900,000$~TPS for batches of
$3,800$~transactions.
However, the throughput decreases slightly after that point, due to
the increased size of compressed batches.
The translation procedure can easily handle two orders of magnitude
more transactions than the throughput of setchain, confirming
hypothesis~\textbf{\HTrans{}}.




\subsection{Experiment Conclusions}
The results of our experiments are: (1) hashing and signing batches of
transactions take significantly less space than compressing batches of
transactions and (2) local procedures are very efficient and do not
affect Setchain throughput.
The reduction in space, which remains constant as
batch size increases, contributes to L1 lower gas consumption.
Since the overhead of our procedures is negligible, our decentralized
arranger does not introduce any significant performance bottlenecks to
Setchain, even if the workload of L2 systems was increased by two
orders of magnitude and there were 256 arranger servers.


%% file: RelatedWork.tex
\section{Related Work}
\label{sec:related-work}

While the limitations of centralized sequencers in rollups are well known, the
development of decentralized sequencers is still in its early
stages~\cite{motepalli2023sok}.
To the best of our knowledge, Metis~\cite{metisDecentralized} is the
only rollup that currently implements a decentralized arranger.
Metis uses Proof-of-Stake and Tendermint~\cite{buchman2016tendermint}
to select a rotating leader.
The leader collects user transactions and generates batches.
Batches are signed using multi-party computation and, when enough
signatures are gathered, the leader posts a signed batch in L1.
However, Metis~\cite{metisDecentralized} does not provide a formal
proof of correctness or mechanisms to check or dispute leader
(mis)behaviour.

Other attempts to implement arrangers include Radius~\cite{Radius},
Espresso~\cite{espressoSequencer} and Astria~\cite{Astria}.
Radius also uses leader election, based on the RAFT
algorithm~\cite{ongaro2014raft}.
Radius servers remain consistent and can reach consensus even when the
leader fails, but~\cite{Radius} does not discuss what happens when
servers refuse to disclose block contents to users.
Radius allows encryption of transactions to prevent MEV, but we do
not address this problem here.
Both Espresso and Astria use one protocol for sequencing transactions
and a different protocol for data availability.
In particular, Espresso uses HotStuff2~\cite{hotstuff} as consensus
mechanism and its own Data Availability layer known as
\emph{tiramisu}, while Astria uses CometBFT~\cite{cometBFT} for
sequencing transactions and leverages Celestia Data Availability
Service~\cite{celestia}.
We decided to tackle both challenges within a single protocol as
Byzantine-resilient solutions do not compose well (see
e.g.~\cite{capretto22setchain}) both in terms of correctness and
efficiency.

Even though Metis, Expresso and Astria use staking to discourage
misconduct among participants, none of these systems provide any
rewards for correct behavior.
Furthermore, to the best of our knowledge, our work is the first to
provide fraud-proof mechanisms usable on L1 to demonstrate the
illegality of batch tags posted by arranger servers.
We also provide an onchain challenge mechanism that can be used to
translate hashes, which has been studied in~\cite{boneh}.
The main difference with~\cite{boneh} is that they assume that the
translation can be directly verified in an L1 smart contract using
SNARK proofs, while we provide an interactive game played over the
Merkle tree nodes.

An extended list of related work can be found in
Appendix~\ref{app:ext-related-work}\footnote{The list of references
  below include works cited in the main body of the paper, and
  \emph{additionally} works cited in the extended related work in
  Appendix~\ref{app:ext-related-work}
  and~\ref{app:decentralizationL2Eth}.}.

%% file: Conclusion.tex
\section{Conclusion}
\label{sec:conclusion}

Rollups aim to improve the scalability of current smart-contract
blockchains, offering a much higher throughput without modifying the
programming logic and interaction of blockchains.
The main element of rollups is \textbf{an arranger} formed by a
sequencer of transactions and a data committee translating hashes
into batches.
Current implementations of arrangers in most rollups are based on
centralized sequencers, either posting compressed batches~(ZK-Rollups
and Optimistic Rollups), or hashes of transactions with a fixed
collection of servers providing data translation~(Validums and
Optimiums).
The resulting systems are fully controlled by a centralized sequencer
server that has full power for transaction reordering and can post
fake blocks without penalties.

In this paper, we formalized the correctness criteria for
arranger services and presented a fully decentralized arranger based
on the efficient Setchain distributed Byzantine tolerant data-type.
We proved our solution correct and reported an empirical
evaluation of its building blocks showing that our solution can
scale far beyond the current demand of existing rollups.

We also described a collection of incentives for servers to behave
correctly when creating, signing and translating batches, and
penalties for posting incorrect information based on fraud-proof
mechanisms.
Using these mechanisms a single honest agent can guarantee that only
valid batches of transactions are executed in L2, regardless of the
number of Byzantine servers in the arranger.



%% file: Appendix/SeqDC.tex
\section{Semi-decentralized Arranger}%
\label{sec:seqDC}


In this section, we describe a \emph{semi-decentralized} distributed
implementation of arrangers consisting of a single process sequencer
and a decentralized DAC (see Fig.~\ref{fig:arranger_seqDC}).
%
%
%
The centralized sequencer implements \<add> locally, create batches
and their hash, inform DAC servers and collect their signatures.
%
DAC servers implement \<translate> mapping the hash into its batch.
\begin{figure}[b!]
    \centering
    \includegraphics[scale=0.36]{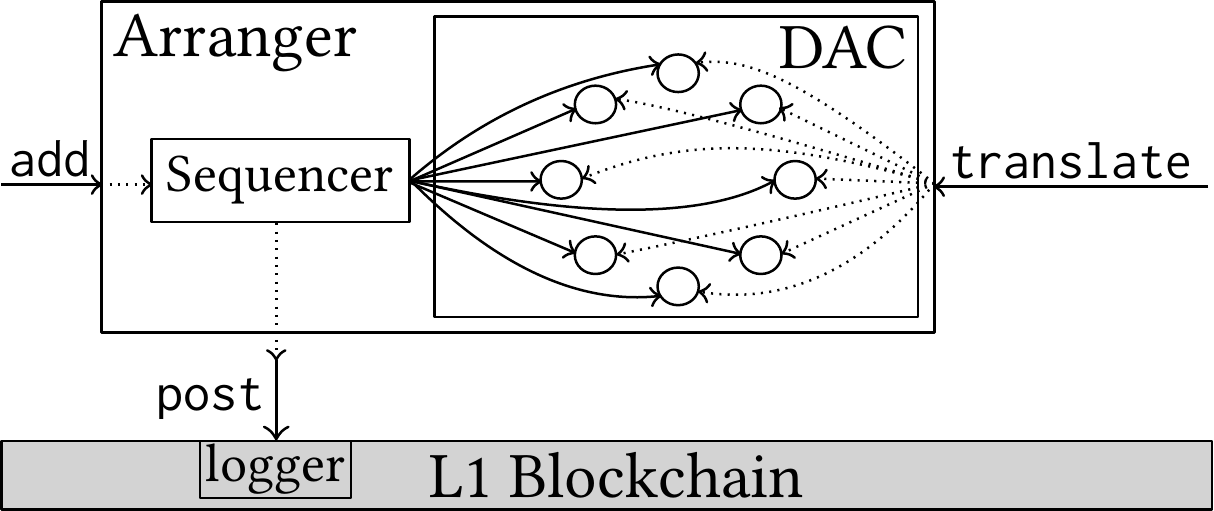}
    \caption{Centralized Sequencer + decentralized DAC.}
    \label{fig:arranger_seqDC}
\end{figure}


%
%
%

\subsection{Implementation}
We provide the pseudo-code implementing sequencers and DAC members
in Alg.~\ref{alg:sequencer} and Alg.~\ref{alg:dc-member},
respectively.
The logger contract is invoked only by the sequencer to store batch
tags.
For simplicity, the sequencer uses a multicast service, to send
messages to all DAC servers (which can be implemented with unicast
channels).
DAC servers answer directly to the sequencer.


Alg.~\ref{alg:sequencer} shows the sequencer pseudo-code.
This algorithm maintains a set \<allTxs> to store all valid added
transactions, a sequence of transactions \<pendingTxs> for
transactions added but not posted, and a natural number \<batchId> as
the identifier of the next batch.
%
Function \<add> adds new valid transactions into \<pendingTxs> and
\<allTxs>.
When there are enough transactions in \<pendingTxs> or a timeout is
reached, a \<timetopost> event is triggered (line~\ref{seq-when}) and
Alg.~\ref{alg:sequencer} generates a new batch with all transactions
in \<pendingTxs> (line~\ref{seq-when-save-pendingTxs}).
Then, the sequencer sends the batch with its identifier and hash to
all DAC servers (line~\ref{seq-send-h}) and waits to receive $f+1$
signatures.
Then, Alg.~\ref{alg:sequencer} aggregates all signatures received,
generates a batch tag and posts it (line~\ref{seq-post}).
Finally, Alg.~\ref{alg:sequencer} increments \<batchId> and cleans
\<pendingTxs> (lines~\ref{seq-increment-batchid}
and~\ref{seq-when-clean-pendingTxs}).
%

\input{algorithms/sequencer}
\input{algorithms/DC-member}

Alg.~\ref{alg:dc-member} shows the pseudo-code of DAC servers,
which maintain a map \<hashes> mapping pairs of identifiers and hashes
$(\<batchId>,h)$ into batches $b$, such that $\hash(b) = h$.
When the sequencer sends requests to sign a batch tag
$(b,\<batchId>,h)$, correct DAC servers check that $h$ corresponds to
$b$ (line~\ref{dc-member-sign-assert}), add the new triplet to
\<hashes>, sign it and send the signature back.
When clients invoke $\<translate>(\<batchId>,h)$, a correct DAC server
returns the value stored in $\<hashes>(\<batchId>,h)$.


%
Clients solve inverse hash resolution by contacting DAC servers and
checking that the batch obtained corresponds to hash $h$, guaranteed
to be provided by a legal server.
%


\subsection{Proof of Correctness}\label{sec:seqDC-props}

We show now that Alg.~\ref{alg:sequencer} and Alg.~\ref{alg:dc-member}
form a correct arranger. 
%
%
We use \<S> for the single sequential server running
Alg.~\ref{alg:sequencer} and \<DAC> to a distributed set of servers
implementing Alg.~\ref{alg:dc-member}.
We assume at most \(f < \frac{n}{2}\) \<DAC> servers are Byzantine
to provide \PrTermination.
%
%

The arranger satisfies \PrValidity and \PrIntegrity, because the
sequencer \<S> is the only process that posts batch tags and we assume
it to be correct.

\begin{lemma}  \label{lem:semiValidity}
  Sequencer \<S> generates legal batch tags of valid
  transactions added by clients. \<S> adds each transaction to
  at most one batch tag.
\end{lemma}
\begin{proof}
  Following Alg.~\ref{alg:sequencer}, \<S> generates batches
  taking transactions from $\<pendingTxs>$, and thus, transactions
  come from clients invocations to \<add>.
  Function \<add> guarantees that each transaction in $\<pendingTxs>$
  is valid and that no transaction is added to
    $\<pendingTxs>$ twice.
  Moreover, $\<pendingTxs>$ is emptied after generating batches so
  transactions are posted only once.
  Before posting batch tags, \<S> gathers at least \(f + 1\)
  signatures, and thus, \<S> posts only legal batch tags.
\end{proof}

The following lemma implies \PrTermination when the
majority of the \<DAC> members are correct.
Note that the sequencer \<S> waits for \(f+1\) answers with
signatures.
%
%
Since there are at least $f+1$ correct \<DAC> servers \<S> will
eventually collect enough signatures.

\begin{lemma}  \label{lem:semiTermination}
  If less than half of the \<DAC> servers are Byzantine then valid
  transactions added through \<add> eventually appear in a legal batch
  tag.
\end{lemma}

\begin{proof}
  Let \(t\) be a valid transaction.
  As a result of \<add>\((t)\), \<S> appends $t$ to $\<pendingTxs>$ where \(t\) remains
  until it is added to a batch.
  At some point after $t$ was added, event $\<timetopost>$ is triggered
  so \<S> generates a new batch of transactions using $\<pendingTxs>$.
  Then, \<S> sends to all \<DAC> members the batch
  (lines~\ref{seq-when}-\ref{seq-send-h}) and waits for $f+1$
  responses (line~\ref{seq-wait}).
  Since $f<n/2$, there are at least $f+1$ correct \<DAC> servers
  executing Alg.~\ref{alg:dc-member}, all of which sign the hash and
  send it back to \<S>
  (lines~\ref{dc-member-receive}-\ref{dc-member-receive-end}).
  Therefore, \<S> eventually receives enough signatures and posts the
  legal batch tag containing $t$ (line~\ref{seq-post}).
\end{proof}

Finally, correct \<DAC> members resolve signed hashes as they only
sign hashes received from \<S> and store their corresponding batches.
Since all legal batch tags are signed by at least one correct DAC
server, it follows that \<S> and \<DAC> satisfy \PrAvailability when
used by correct clients.

\begin{lemma}\label{lem:semiAvailability}
  Correct \<DAC> members resolve signed hashes.
\end{lemma}

\begin{proof}
  Let $m$ be a correct \<DAC> member.
  Following Alg.~\ref{alg:dc-member}, \(m\) only signs hashes
  answering requests from \<S>.
  Upon receiving a request from \<S> to sign $(h,b,\<batchId>)$, after
  checking that $h$ is correct and \(\<batchId>\) has not been
  used, \(m\) saves the data received, signs the request and answers
  to \<S>.
  As no entry in $\<hashes>$ is ever deleted, whenever
  $m.$\<translate>$(\<batchId>,h)$ is invoked, $m$ uses map
  $\<hashes>$ to find batch $b$ and returns it.
\end{proof}

From the previous lemmas, we have Alg.~\ref{alg:sequencer} plus
Alg.~\ref{alg:dc-member} guarantee that \<S> and \<DAC> form a correct
arranger.
%
\begin{theorem}
  Clients contacting \<S> to add transactions and a \<DAC> member
running Alg.~\ref{alg:dc-member} signing batches satisfy all correct
arranger properties.
\end{theorem}

%% file: algorithms/sequencer.tex
\begin{figure}[t!]
  \begin{algorithm}[H]
    \caption{\small Sequencer.}
    \label{alg:sequencer}
    \small
    \begin{algorithmic}[1]
  	\State \textbf{init:} $\<allTxs> \leftarrow \emptyset $,
        $\<pendingTxs> \leftarrow \epsilon$, $\<batchId> \leftarrow 0$~\label{seq-init}
        \Function{\textnormal{\<add>}}{$t$} \label{seq-add}
          \State \<assert> $\<valid>(t) $ and $t \notin \<allTxs>$ \label{seq-add-assert}
          \State $\<pendingTxs> \leftarrow \langle \<pendingTxs>, t \rangle$ \label{seq-add-pendingTxs}
          \State $\<allTxs> \leftarrow \<allTxs> \cup \{t\}$ \label{seq-add-allTxs}
          \State \Return
        \EndFunction \label{seq-add-end}
        \Upon{\<timetopost>}\label{seq-when} 
          \State $b \leftarrow \<pendingTxs>$ \label{seq-when-save-pendingTxs}
          \State $h \leftarrow \hash(b)$ \label{seq-when-hash-b}
          \State \texttt{\textbf{DC}.Multicast}$(\<signReq>(b,\<batchId>,h))$~\label{seq-send-h}
          \State $\<wait>$ for $f+1$: $\<signResp>(\<batchId>,h,i,sig_i(\<batchId>,h))$ \label{seq-wait}
          \State \<logger>.\<post>$(\<batchId>, h, \<BLSAggregate>(\bigcup_i sig_i(\<batchId>,h)))$ \label{seq-post}
         \State $\<batchId> \leftarrow \<batchId> +
         1$ \label{seq-increment-batchid}
         \State $\<pendingTxs> \leftarrow \epsilon$ \label{seq-when-clean-pendingTxs}
        \EndUpon \label{seq-when-end}
    \end{algorithmic}
  \end{algorithm}
  \vspace{-2em}
\end{figure}


%% file: algorithms/dc-member.tex
\begin{figure}[t!]
  \begin{algorithm}[H]
    \caption{\small Correct DAC server $i$.}
    \label{alg:dc-member}
    \small
    \begin{algorithmic}[1]
  	\State  \textbf{init:} $\<hashes> \leftarrow \emptyset$
        \Function{\textnormal{\<translate>}}{$\textit{id},h$}
           \If{$(\<id>,\_)\notin \<hashes>$}
             \<return> \<invalidId>
           \Else
             \If{$(\<id>,h)\notin \<hashes>$}
                \<return> \<invalidHash>
             \Else
               ~\<return> $\<hashes>(\<id>,h)$
             \EndIf
           \EndIf  
        \EndFunction
        \Receive{$\<signReq>( b,\<id>, h)$ from Sequencer} \label{dc-member-receive}
          \State \<assert> $h == \hash(b)$ and $(\<id>,\_) \notin \<hashes>$ \label{dc-member-sign-assert}
          \State $\<hashes> \leftarrow \<hashes> \cup \{((\<id>,h),b)\}$ \label{dc-member-sign-add}
          \State \textbf{send} $\<signResp>(\<id>,h,i, sig_i(\<id>,h))$ to
          Sequencer \label{dc-member-sign-send}
        \EndReceive \label{dc-member-receive-end}
    \end{algorithmic}
  \end{algorithm}
  \vspace{-2em}
\end{figure}


%% file: Appendix/SetchainImpl.tex
\section{Setchain Server Implementation}\label{sec:setchain-implementation}\label{app:setchain}

An overview of a setchain server implementation can be found below.

\input{algorithms/setchain}

%% file: algorithms/setchain.tex
  \begin{algorithm}[H]
  \caption{\small Server implementation of Setchain.}~\label{alg:setchain}
   \small
   \begin{algorithmic}[1]
     \State  \textbf{Init:} $\EPOCH \leftarrow 0,$
     \hspace{2em} $\HISTORY \leftarrow \emptyset$ 
     \State  \textbf{Init:} $\THESET \leftarrow \emptyset$\label{alg3:theset}
     \Function{\<Get>}{~}
        \State \<return> $(\THESET, \HISTORY,\<epoch>)$\label{alg3:get}
     \EndFunction
     \Function{\<Add>}{$e$} \label{alg3-add-begin}
        \State \<assert> $valid(e)$ and $e \notin \THESET$
        \State add element $e$ to $\<theset>$
        \State propogate $e$
        \If {$e = epoch\_signature(h,E,sign(hash(E)))$}
           \State $\<trigger>$ $\<added>(epoch\_signature(h,E,sign(hash(E))))$
        \EndIf
      \EndFunction  
      \Function{\<EpochInc>}{$h$} \label{alg3-epochinc-begin}
         \State \<assert> $h==\EPOCH +1$
         \State collaboratively compute new epoch $E$ \Comment{triggers event
           $\<newepoch>(h,E)$}
         \State validate elements in $E$
	 \State add new valid elements in $E$ to \<theset>
         \State add new epoch $E$ to \<history> 
         \State increment $\<epoch>$ 
         \State sign the new epoch and add
         $epoch\_signature(h,E,sign(hash(E))))$ to the setchain \label{alg:setchain-sign-epoch}
         \State $\<trigger>$ $\<newepoch>(h,E)$
      \EndFunction  
    \end{algorithmic}
  \end{algorithm}\vspace{-2em}

%% file: Appendix/htlc.tex
\vspace{20em}
\section{Payment system}\label{app:htlc}

We present the pseudo-code implementing a payment system based on
hashed-timelock contract for clients described in
Section~\ref{sec:incentives:translate}.
\begin{lstlisting}[language=Solidity,numbers=none]
  contract Payment {
    address owner, beneficiary;
    bytes32 secret;
    uint deadline;

    constructor(address _beneficiary, bytes32 _secret, uint _deadline){
      beneficiary = _beneficiary;
      secret = _secret;
      deadline = _deadline;
    }
    
    function claim(bytes k) public {
      require(msg.sender == beneficiary);
      require(SHA256(k) == secret);
      beneficiary.transfer(balance); 
    }

    function withdraw() public {
      require(msg.sender == owner);
      require(now > deadline);
      owner.transfer(balance);
    }
  }
\end{lstlisting}


%% file: Appendix/Challenges-alg.tex
\section{Challenges pseudo-code}
\label{app:challenges-alg}

We present a pseudo-code for the L1 smart contract
arbitrating the challenges described in
Section~\ref{sec:incentives:challenges}.
These smart contracts are used during a challenge game that implements
the fraud-proof mechanisms described in
Section~\ref{sec:incentives:challenges}.
An honest challenger \Att uses the methods of this contract to play
the game described in Section~\ref{sec:incentives:challenges}, with a
guarantee that if \Att is correct \Att will win the game.
Absence of a move within a pre-defined interval results in the
corresponding player loosing the game.
Note that none of this contracts has a loop or invokes remotes
contracts, so termination arguments (and gas calculations) can be
easily performed.

\input{algorithms/dataChallenge}
\input{algorithms/signatureChallenge}
\input{algorithms/validityChallenge}
\input{algorithms/integrityChallenge}


%% file: algorithms/dataChallenge.tex
\begin{figure}
  \begin{algorithm}[H]
    \caption{\small Data Availability Challenge.}
    \label{alg:dataChallenge}
    \small
    \begin{algorithmic}[1]        
    \Function{InitDataChallenge}{$\id,\merkleRoot,\sigs,\staker,\val,\sender$}\label{initDataChallenge-begin}
          \State \<assert> $(\id,\merkleRoot,\sigs) \in \proposedBatches$
          \State \<assert> $\staker \in \stakers[(\id,\merkleRoot,\sigs)]$
          \State \<assert> $\val \geq \CC{data}$

          \State $\challengeId \leftarrow \texttt{createChallenge}(\sender,
          \staker, \id, \merkleRoot, \sigs, \texttt{data})$
          \State $\challengedNodes[challengeId][1] \leftarrow
          (\texttt{hash}(\merkleRoot), \texttt{size}(\merkleRoot))$
          \State $\canBeChallenged[\challengeId] \leftarrow \emptyset$
        \EndFunction \label{initDataChallenge-end}
        \Function{PresentDataCompressed}{$\challengeId, n, b_l, h_l,
          b_r, h_r,\sender$}
          \State \<assert> $\sender == \texttt{currentPlayer}(\challengeId)$
          \State \<assert> $n \in \challengedNodes[\challengeId]$
          \State $(h_n,  size_n) \leftarrow \challengedNodes[\challengeId][n]$
          \State \<assert> $size_n > 1$
          \If {$\hash(h_l ++ h_r) \neq h_n$}
            \State $\texttt{setWinner}(\challengeId,\texttt{nextPlayer}(\challengeId))$
          \Else
            \State
            $\canBeChallenged[\challengeId][\texttt{leftChild}(n)]
            \leftarrow h_l, \texttt{sizeLeftChild}(size_n)$
            \State
            $\canBeChallenged[\challengeId][\texttt{rightChild}(n)]
            \leftarrow  (h_r, \texttt{sizeRightChild}(size_n))$
            \State $\challengedNodes[\challengeId] \leftarrow
            \texttt{remove}(challengedNodes[challengeId], n)$
            \If{$\challengedNodes[\challengeId] == \emptyset$}  
              \State $\texttt{changeTurn}(\challengeId)$  
            \EndIf
          \EndIf    
        \EndFunction
        \Function{PresentElement}{$\challengeId, n, e,\sender$} \label{presentElement-begin}
          \State \<assert> $\sender == \texttt{currentPlayer}(\challengeId)$
          \State \<assert> $n \in \challengedNodes[\challengeId]$
          \State $(h_n, size_n) \leftarrow
          \challengedNodes[\challengeId][n]$
          \State \<assert> $size_n == 1$
          \If{$\hash(e) != h_n$}
             \State $\texttt{setWinner}(\challengeId,\texttt{nextPlayer}(\challengeId))$ 
          \Else   
            \State $\challengedNodes[\challengeId] \leftarrow
             \texttt{remove}(\challengedNodes[\challengeId], n)$
            \If{$\challengedNodes[\challengeId] == \emptyset$}  
              \State $\texttt{changeTurn}(\challengeId)$  
            \EndIf
          \EndIf  
        \EndFunction \label{presentElement-end} 
        \Function{ChallengePosition}{$\challengeId, n, \moreChallenges,\sender$}
          \State \<assert> $\sender == \texttt{currentPlayer}(\challengeId)$
          \State \<assert> $n \in \canBeChallenged[\challengeId]$
          \State $\challengedNodes[\challengeId][n] = \canBeChallenged[\challengeId][n]$
          \State $\canBeChallenged[\challengeId] \leftarrow
          remove(canBeChallenged[\challengeId], n)$
          \If {$!moreChallenges$}
            \State $\texttt{changeTurn}(\challengeId)$
          \EndIf  
         \EndFunction
         \Function{Timeout}{$\challengeId$}
          \If{$\texttt{runoutOfTime}(\challengeId)$}
             \State $\texttt{setWinner}(\challengeId,\texttt{nextPlayer}(\challengeId))$ 
          \EndIf  
         \EndFunction   
    \end{algorithmic}
  \end{algorithm}
  \vspace{-2em}
\end{figure}


%% file: algorithms/signatureChallenge.tex
\begin{figure}
  \begin{algorithm}[H]
    \caption{\small Signature Challenge.}
    \label{alg:signatureChallenge}
    \small
    \begin{algorithmic}[1]
        \Function{SignatureChallengeSize}{$\id,h,\sigs,\val,\sender$}
          \State \<assert> $(\id,h,\sigs) \in \proposedBatches$
          \State \<assert> $\val >= \CC{singature}$ 
          \If{$|sigs| < f+1$}
            \State $\texttt{removeAllStakes}(id,h,sigs,sender,\texttt{signature})$
           \EndIf 
         \EndFunction
         \Function{SignatureChallengeAggregation}{$\id,h,\sigs,\val,\sender$}
          \State \<assert> $(\id,h,\sigs) \in \proposedBatches$
          \State \<assert> $\val >= \CC{singature}$ 
          \State $\aggPKs \leftarrow \empty$
          \For{$i \in \sigs$}
            \State $\aggPKs \leftarrow
            \texttt{aggregatePK}(\aggPKs, \KWD{arrangerPK}[i])$
          \EndFor
          \If{$! validSignature(\aggPKs, (\id,h), \sigs)$}
            \State $\texttt{removeAllStakes}(id,h,sigs,sender,\texttt{signature})$
           \EndIf 
        \EndFunction
    \end{algorithmic}
  \end{algorithm}
  \vspace{-2em}
\end{figure}


%% file: algorithms/validityChallenge.tex
\begin{figure}
  \begin{algorithm}[H]
    \caption{\small Validity Challenge.}
    \label{alg:validityChallenge}
    \small
    \begin{algorithmic}[1]        
        \Function{InitValidityChallenge}{$id,h,sigs,e,h_e,ePos,midH,staker,value,sender$}
          \State \<assert> $(id,h,sigs) \in proposedBatches$
          \State \<assert> $staker \in stakers[(id,h,sigs)]$
          \State \<assert> $value >= cc_{validity}$
           \State \<assert> $\hash(e) == h_e$ 
           \State \<assert> $\texttt{valid}(e)$
           
          \State $challengeId \leftarrow \texttt{createChallenge}(sender,
          staker, id, h, sigs,\texttt{validity})$
          \State $challengeState[challengeId] \leftarrow \texttt{Bisection}$
          \State $bottom[challengeId] \leftarrow (ePos, h_e)$
          \State $top[challengeId] \leftarrow (1, h)$
          \State $mid[challengeId] \leftarrow (\texttt{midPos}(1, ePos), midH)$
        \EndFunction
        \Function{BisectionSelectSubpath}{$challengeId, top,sender$}
          \State \<assert> $challengeState[challengeId] == \texttt{Bisection}$
          \State \<assert> $sender == \texttt{currentPlayer}(challlengeId)$
          \If {$top$}
            \State $bottom[challengeId] \leftarrow mid[challengeId]$
          \Else
            \State $top[challengeId] \leftarrow mid[challengeId]$ 
          \EndIf
          \If {$\texttt{dist}(bottom[challengeId].pos, top[challengeId].pos) ==
            1$}
            \State $challengeState[challengeId] \leftarrow
            \texttt{OneStep}$
         \EndIf  
         \State $\texttt{changeTurn}(challengeId)$  
        \EndFunction
        \Function{BisectionMembershipProof}{$challengeId, h,sender$}
          \State \<assert> $challengeState[challengeId] == \texttt{Bisection}$
          \State \<assert> $sender == \texttt{currentPlayer}(challlengeId)$
          \State $mid[challengeId] =
          (\texttt{midPos}(top[challengeId].pos, bottom[challengeId].pos), h)$
        \EndFunction
        \Function{OneStepMembershipProof}{$challengeId, h,sender$}
          \State \<assert> $challengeState[challengeId] == OneStep$
          \State \<assert> $sender == \texttt{currentPlayer}(challlengeId)$
          \State $c \leftarrow \texttt{concatenate}(bottom[challengeId].hash,
          h, bottom[challengeId].pos)$
          \If{$\hash(c) == top[challengeId].hash$}
            \State $\texttt{setWinner}(challengeId,\texttt{currentPlayer}(challengeId))$ 
          \Else
            \State $\texttt{setWinner}(challengeId,\texttt{nextPlayer}(challengeId))$
          \EndIf  
         \EndFunction
         \Function{Timeout}{$challengeId,sender$}
           \Comment{Same as in Alg.~\ref{alg:dataChallenge}}
         \EndFunction   
    \end{algorithmic}
  \end{algorithm}
  \vspace{-2em}
\end{figure}


%% file: algorithms/integrityChallenge.tex
\begin{figure}
  \begin{algorithm}[H]
    \caption{\small Integrity Challenge.}
    \label{alg:integrityChallenge}
    \small
    \begin{algorithmic}[1]        
        \Function{InitIntegrityChallenge}{$id0,h0,sigs0,e,eH,ePos0,id1,ePos1,
          staker, value, sender$}
          \State \<assert> $(id1,h1,sigs1) \in proposedBatches$
          \State \<assert> $staker \in stakers[(id,h,sigs)]$
          \State \<assert> $value >= cc_{integrity}$
  
          \State \<assert> $\hash(e) = eH$ 
          \If {$id1 == id0$}
            \State $h1 \leftarrow h0$
          \Else
            \State \<assert> $id1 \in \confirmedBatches$
            \State $h1 \leftarrow \confirmedBatches[id1].hash$
          \EndIf  
          \State $\mathit{challengeId} \leftarrow \texttt{createChalenge}(sender,
          staker, id0, h0, sigs0, \texttt{integrity})$
          \State $challengeState[challengeId] \leftarrow \texttt{PathSelection}$
          \State $Path[0][challengeId] \leftarrow (eH,ePos0, h0)$
          \State $Path[1][challengeId] \leftarrow (eH,ePos1, h1)$
        \EndFunction
        \Function{SelectPath}{$challengeId, pathId, sender$}
          \State \<assert> $challengeState[challengeId] == \texttt{PathSelection}$
          \State \<assert> $sender == currentPlayer(challlengeId)$
          \State $(bottomH, bottomPos, topH) \leftarrow Path[pathId][challengeId]$
          \State $bottom[challengeId] \leftarrow (bottomPos,bottomH) $
          \State $top[challengeId] \leftarrow (1,topH)$ 
          \State $mid[challengeId] = (midPos(bottomPos, 1),
          midH) $
          \State $challengeState[challengeId] \leftarrow \texttt{Bisection}$
          \State $\texttt{changeTurn}(challengeId)$  
        \EndFunction
         \Function{BisectionSelectSubpath}{$challengeId$}
           \Comment{Same as in Alg.~\ref{alg:validityChallenge}}
         \EndFunction
         \Function{BisectionMembershipProof}{$challengeId$}
           \Comment{Same as in Alg.~\ref{alg:validityChallenge}}
         \EndFunction
        \Function{OneStepMembershipProof}{$challengeId$}
           \Comment{Same as in Alg.~\ref{alg:validityChallenge}}
         \EndFunction
        \Function{Timeout}{$challengeId$}
           \Comment{Same as in Alg.~\ref{alg:dataChallenge}}
         \EndFunction
    \end{algorithmic}
  \end{algorithm}
  \vspace{-2em}
\end{figure}


%% file: Appendix/ProofIncentives.tex
\section{Fraud-proof Mechanisms and Incentives - Proofs}%
\label{app:ProofIncentives}

We now provide the proof of the propositions presented in
Section~\ref{sec:incentives}.

\repeatatproposition{DAC}
\begin{proof}
  Let \Att be an honest agent, \(t\) an unconfirmed signed batch tag
  posted in L1, \Stt a staker in \(t\), and \(mt\) the Merkle Tree
  generated by the batch of transaction in \(t\).
  We will prove that if \Att and \Stt play the data availability
  challenge, then by its end either \Att learn all transactions in
  \(t\) or \Stt loses its stake.
  Since \Att can challenge all \(t\) stakers until either \Att learns
  all transactions or there are no more stakers in \(t\), the
  proposition follows.

  We will prove by structural induction over the Merkle tree \(mt\)
  that if a node \(n\) hash is posted in the L1 contract, then \Att
  can learn all transactions corresponding to that node subtree, or
  \Stt gets its stake removed.
  
  The base case is when \(n\) is a leaf.
  If \Att does not know the transaction corresponding to \(n\) yet,
  \Att challenges \Stt to reveal it.
  \Stt can either lose its stake or answer the challenge correctly by
  posting the transaction in the L1 smart contract that checks that
  its hash matches \(n\) hash.

  For the inductive step, we consider that \(n\) has a left child \(l\)
  and a right child \(r\).
  The inductive hypothesis is that if the hash of \(l\) and \(r\) are
  posted in the L1 smart contract, then \Att can learn all transactions
  corresponding to each subtree, or \Stt gets its stake removed.
  If \Att does not know the transactions in \(n\) subtree
  yet, \Att challenges \Stt to reveal them.
  \Stt can either lose its stake or answer the challenge correctly by
  posting the nodes \(l\) and \(r\) hashes in the L1 smart contract
  that checks that its concatenation hashes to \(n\) hash.
  \Stt can also post the compressed version of the batches of
  transactions corresponding to \(l\) and \(r\) subtrees, in which
  case \Att learns all transaction in the \(n\) subtree, by
  concatenating both batches.
  Otherwise, since \(l\) and \(r\) hashes are posted in L1, \Att can
  challenge each of them and by inductive hypothesis, either \Stt gets
  its stake removed or \Att learns all transactions in both subtrees,
  and therefore all transactions in \(n\) subtree.
\end{proof}

\repeatatproposition{correctness}
\begin{proof}
    The proof depend on which property \(t\) violates:
  \begin{itemize}
    \item \textbf{Violation of B1:} \(t\) is not signed by enough
      arranger processes or the signature is not valid.
      In either case, \Att can invoke the corresponding function in
      the L1 smart contract which will verify the violation resulting
      in \(t\) being discarded.
    \item \textbf{Violation of B2:} There is an element \(e\) in \(t\)
      that is not a valid a transaction.
      
      Let \Stt by in a staker in \(t\)
      We will prove that \Att can win the validity challenge against
      \Stt.
      
      First, \Att posts \(e\) in the L1 smart contract which verifies
      that it is invalid.
      Since, \Att knows all transactions in \(t\) \Att can reconstruct
      its Merkle tree \(mt\).
      \Att also posts \(e\) position in \(mt\), \(e\) hash, and the
      hash of the node in the middle of path from \(e\) leaf to \(mt\)
      root.
      
      To answer each challenge from \Stt, \Att must either provide one
      of the hashes in the path from \(e\) to \(mt\) root, or, when
      the bisection process is over, one of the children of a node in
      that path.
      Since \Att know all hashes in \(mt\),\Att can successfully
      answer all \Stt challenge and win the validity challenge.

      Therefore, \Att can play the validity challenge against each of
      \(t\) stakers and win every time.
      As consequence, all stakers loose their stakes and \(t\) is
      discarded.
      
    \item \textbf{Violation of B3 and B4:} Note that \Att know not
      only all transactions in \(t\) but also all previous batches.
      Therefore, if a transaction is duplicated, \Att can compute the
      Merkle tree(s) where it appears.
      The proof for these cases is analogous to the case where \(t\) violates B2.
  \end{itemize}
\end{proof}

\repeatatproposition{correctness2}

\begin{proof}
  Since \Att knows all transactions in \(t\), it can compute their
  effect and post it as an L2 block.
  If the L2 block is challenged, \Att can win the challenge
  by correctly playing the fraud proof mechanism of the L2 blockchain.

  The remaining step is to show that \Att can keep its stake in \(t\)
  until the corresponding L2 block is confirmed.
  We demonstrate that \Att can win each type of challenge separately:
  \begin{itemize}
  \item \textbf{data challenge}: \Att computes the Merkle tree \(mt\)
    associated with the transactions in \(t\) and uses \(mt\) to
    answer all challenges correctly.
  \item \textbf{signature challenge}: As \(t\) is a legally signed
    batch tag, it is correctly signed by at least \(f+1\) arranger
    processes.
    In this case, the L1 smart contract will fail the challenge.
  \item \textbf{validity challenge}: A challenger might post in L1 an
    element \(e\), along with its hash, its claimed position in \(t\)
    Merkle tree \(mt\), and the claimed hash of the middle node in the
    path from \(e\) to \(mt\) root.
    Element \(e\) must be invalid, otherwise \Att immediately wins the
    challenge as the L1 smart contract verifies that the element is
    invalid.
   Since, \(t\) is legal, the element does not belong to it.
      
   To win the challenge, \Att maintains the invariant that after its
   turn in the challenged subpath the top node belongs to \(mt\) while
   the bottom node does not.
   The path from the \(e\) to the root of \(mt\) satisfies this
   invariant.
   In each turn \Att selects the top subpath if the middle node
   provided by the challenger does not belong to \(mt\), and the
   bottom path otherwise.
   This strategy clearly preserve the invariant.
      
   When the challenged subpath has two nodes with hashes \(h_p\) and
   \(h_c\), where the node with hash \(h_p\) is the parent of the node
   with hash \(h_c\), the challenger needs to post a hash \(h\) such
   that the concatenation of \(h\) with \(h_c\) hashes to \(h_p\).
      Such \(h\) proves that \(h_c\) and \(h_p\) belong to the same
      Merkle Tree.
      However, the challenger cannot provide such an \(h\) as, by the
      invariant, \(h_p\) belongs to \(mt\) but \(h_c\) does not.
    \item \textbf{integrity challenge}:
      Since \(t\) is legal, no transaction in \(t\) is duplicated.
      Therefore, one of the claimed path by the challenger must be
      incorrect.
      Note that \Att knows not only all transactions in \(t\) but also
      all previous batches, and therefore if there is an integrity
      challenge it can compute the Merkle tree of all involved batches
      and detect the incorrect path.
      Then the proof is analogous to the validity challenge.
    \end{itemize}

    In conclusion, \Att can win all challenges and guarantee the
    confirmation of \(t\).
\end{proof}

\repeatatproposition{cost}

\begin{proof}
  We analyze the minimum budget \(b\) based on the violations of
  \(t\):
\begin{compactitem}
\item If \(t\) only violates B1, \Att needs to play the signature
  challenge, which requires a $b \geq \CC{signature}$.
\item If \(t\) only violates B2, then \(\CC{data} + \CC{validity}\)
  tokens are necessary and sufficient.
  Since \(t\) violates B2, one of its transaction is invalid.
  Therefore, \Att needs to know the transaction in \(t\).
  As all stakers in \(t\) can refuse to participate in the translation
  protocol, \Att might need to play the data challenge, which requires
  at least \(\CC{data}\) tokens.
  \Att can learn all transactions in \(t\) the first it plays a data
  challenge, losing \(\CC{data}\) tokens.
  Once \Att learned all transactions, \Att needs at least
  \(\CC{validity}\) tokens to play the validity challenge and make
  stakers lose their stake.
  Then, \Att needs at least \(\CC{data} + \CC{validity}\) tokens.
  To see that \(\CC{data} + \CC{validity}\) tokens is sufficient, \Att
  can first (1) challenge one by one all \(t\) stakers to the data
  challenge until either \Att learns all transactions in \(t\) or
  there no more stakers in \(t\) and it is discarded, and then (2)
  challenge remaining stakers one by one to the validity
  challenge.
  For the first phase, \(\CC{data}\) tokens are enough.
  Every time \Att wins a data challenge its budget increases because
  \Att spends \(\CC{data}\) tokens but receives \(\CR{data}\) tokens,
  and \(\CR{data} > \CC{data}\).
  Then, if \Att starts with at least \(\CC{data}\) and \Att wins the
  data challenge, \Att will have enough token to continue playing.
  Similarly, \Att reward for winning the validity challenge is larger
  than the cost of playing the challenge, then \(\CC{validity}\)
  tokens are enough to complete the second phase.
  Then, to complete both phases $\CC{data} + \CC{validity}$ tokens are
  sufficient.
\item If \(t\) only violates (B3) or (B4), \Att must play the
  integrity challenge but needs to learn all transactions in \(t\)
  before.
  With a reasoning analogous to the validity challenge, the minimum
  budget needed in this case is
  $b \geq \CC{data} + \CC{integrity}$
\item If a batch violates multiple properties, \Att chooses the
  cheapest way to remove it.
\end{compactitem}

Consequently, \Att can remove any illegal batch tag with a budget of
at least
\(max(\CC{signature}, \CC{data} + max(\CC{validity},\CC{integrity}))\)
tokens.
\end{proof}

%% file: Appendix/DecentralizationL2Eth.tex
\section{Decentralization of Ethereum Layer 2 scaling systems}
\label{app:decentralizationL2Eth}


 \begin{table}[h]
     \centering
       \begin{tabular}{|l|l|l|}
        \hline
        \textbf{Project} & \textbf{Type of Rollup} & \textbf{Arranger
                                                     * }
        \\
        \hline
        Arbitrum One~\cite{ArbitrumNitro} & Optimistic Rollup & Centralized \\
        Optimism~\cite{optimism} & Optimistic Rollup & Centralized \\
        Base~\cite{base} & Optimistic Rollup & Centralized \\
        Blast~\cite{blast} & Optimistic Rollup & Centralized \\
        Mantle~\cite{mantle} & Optimium & Centralized \\
        Linea~\cite{linea} & ZK Rollup & Centralized \\
        Starknet~\cite{starknet} & ZK Rollup & Centralized \\
        zkSync Era~\cite{zksyncera} & ZK Rollup & Centralized \\
        Mode Network~\cite{mode} & Optimistic Rollup & Centralized \\
        Manta Pacific~\cite{manta} & Optimium & Centralized \\
        Metis Andromeda~\cite{metis} & Optimium & Decentralized \\
        Scroll~\cite{scroll} & ZK Rollup & Centralized \\
        dYdX v3~\cite{dydx} & ZK Rollup & Centralized \\
        Immutable X~\cite{immutablex} & Validium & Centralized \\        
        RSS3~\cite{rss3} & Optimistic Rollup & Centralized \\
        Polygon zkEVM~\cite{polygonzkEVM} & ZK Rollup & Centralized \\
        Astar zkEVM~\cite{astar} & Validium & Centralized \\
        X Layer~\cite{xlayer} & Validium & Centralized \\
        BOB~\cite{bob} & Optimistic Rollup & Centralized \\
        Loopring~\cite{loopring} & ZK Rollup & Centralized \\
        Karak~\cite{karak} & Optimium & Centralized \\
        Fraxtal~\cite{fraxtal} & Optimium & Centralized \\
        zkSync Lite~\cite{zksynclite} & ZK Rollup & Centralized \\
        ApeX~\cite{apex} & Validium & Centralized \\
        Reya~\cite{reya} & Validium & Centralized \\
        Aevo~\cite{aevo} & Optimium & Centralized \\
        DeGate V1~\cite{degate} & ZK Rollup & Centralized \\
        ZKFair~\cite{zkfair} & Validium & Centralized \\
        Arbitrum Nova~\cite{ArbitrumNitro} & Optimium & Centralized \\
        rhino.fi~\cite{rhinofi} & Validium & Centralized \\
        Zora~\cite{zora} & Optimistic Rollup & Centralized \\
        Sorare~\cite{sorare} & Validium & Centralized \\
        Kinto~\cite{kinto} & Optimistic Rollup & Centralized \\
        Boba Network~\cite{boba} & Optimistic Rollup & Centralized \\
        Lyra~\cite{lyra} & Optimium & Centralized \\
        Kroma~\cite{kroma} & Optimistic Rollup & Centralized \\
        Paradex~\cite{paradex} & ZK Rollup & Centralized \\
        ZKSpace~\cite{zkspace} & ZK Rollup & Centralized \\ 
        Redstone~\cite{redstone} & Optimium & Centralized \\
        tanX~\cite{tanX} & Validium & Centralized \\
        PGN~\cite{pgn} & Optimium & Centralized \\
        Ancient8~\cite{ancient8} & Optimium & Centralized \\ 
        Orderly Network~\cite{orderly} & Optimium & Centralized \\
        Cyber~\cite{cyber} & Optimium & Centralized \\
        Myria\cite{myria} & Validium & Centralized \\
        Parallel~\cite{parallel} & Optimistic Rollup & Centralized \\
        Cartesi~\cite{cartesi} & Optimistic Rollup & Centralized \\    
        Hypr~\cite{hypr} & Optimium & Centralized \\
        Metal~\cite{metal} & Optimistic Rollup & Centralized \\
        ReddioEx~\cite{reddioex} & Validium & Centralized \\
        Fuel V1~\cite{fuelv1} & Optimistic Rollup & Centralized \\    
        \hline
    \end{tabular}
\caption{Decentralization status of Layer2 solutions on top of
      Ethereum. Data from~\cite{l2beat}}
    \label{tab:decentralization-layer2eth}

\end{table}

* The decentralization status of the arranger for each project X was
obtained from \url{https://l2beat.com/scaling/projects/X#operator}
(accessed: 2024-05-22)


%% file: Appendix/ExtendedRelatedWork.tex
\section{Extended Related Work}
\label{app:ext-related-work}


%
Various approaches were explored to enhance blockchain
scalability, including:
sharding~\cite{Zamani2018RapidChain, Dang2019Sharding,
  das2020efficient, kokoris2018omniledger,rana2022free2shard,
  shen2023gt, cai2022benzene},
faster consensus algorithms~\cite{Wang2019FastChain, wang2020prism,
  jia2022themis},
parallelism~\cite{lovejoyparsec, liu2021parallel},
application-specific blockchains with Inter-Blockchain Communication
capabilities~\cite{wood2016polkadot,kwon2019cosmos,zamyatin2019xclaim},
state channels
\cite{Poon2016lightning,Coleman2015StateChannels,Dziembowski2017PERUNVP,
  miller2017sprites,bentov2017instantaneous, frassetto2022pose},
sidechains \cite{nick2020liquid, mallaki2022off, polygonPoS},
childchains \cite{poon2017plasma},
combination with Trusted Executed Environments
\cite{brandenburger2018blockchain, cheng2019ekiden, lind2019teechain,
  zhang2019paralysis, das2019fastkitten, xiao2020privacyguard,
  cai2022benzene},
zero-knowledge
Rollups~\cite{linea, starknet, zksyncera,scroll,dydx,polygonzkEVM,
  loopring, zksynclite,degate,paradex,zkspace},
Optimistic Rollups~\cite{ArbitrumNitro, optimism,
  base, blast, mode,rss3,bob,zora,kinto,boba,kroma,parallel,cartesi,metal,fuelv1}, Validiums~\cite{immutablex,astar,xlayer,apex,reya,zkfair,rhinofi,sorare,tanX,myria,reddioex}
and Optimiums~\cite{mantle, manta,metis,karak,fraxtal, aevo,lyra,redstone,pgn,ancient8,orderly,cyber,hypr}.
For an in-depth exploration of these approaches,
see~\cite{croman2016scaling, hafid2020scaling, hewa2021survey,
  zhou2020solutions, sanka2021systematic,thibault2022blockchain,
  unnikrishnan2022survey, zhou2023overview, gudgeon2020sok,
  nasir2022scalable, zhou2020solutions}.

We focused our research on optimistic rollups solutions and study their
decentralization.
No L2 solution as today (including zero-knowledge rollups) operating on Ethereum
is fully decentralized (see Appendix~\ref{app:decentralizationL2Eth}).
Our decentralized arranger applies to Arbitrum as well as other
L2 solutions, like Optimism~\cite{optimism}.
In Optimism, their centralized sequencer posts batches of transactions
and a list describing the new L2 state after executing each
transaction.
This would require to adjusts our arranger serves to generate lists of
L2 states, sign these lists, and incorporate them in the batch tags.


\subsubsection*{Data availability.}

A challenge in data availability is supporting lightweight clients
that only test data availability for hashes posted by arrangers,
without downloading complete batches, as explored
in~\cite{albassam2019fraud}.
However, we do not anticipate a use with such lightweight clients, since STFs in
our framework require complete batches to compute new L2 states.
Additionally, we provide a mechanism
(Section~\ref{sec:incentives:challenges}) to ensure data availability
on L1 when direct server contact fails.
Furthermore, once STFs successfully acquire the necessary data,
compute the new L2 state and the block is confirmed, this confirmation
becomes irreversible, even if the batch data used for computation
becomes unavailable.

\subsubsection*{Incentives in blockchain}
\imarga{maybe add about incentives in other offchain approachs}
Numerous blockchain protocols incorporate economic incentives to
motivate participants to behave correctly.
These incentives involve rewarding correct behavior, as seen in
proof-of-work protocols~\cite{jakobsson1999proofs, dwork1992pricing},
or penalties for incorrect behavior, common in
proof-of-stake~\cite{saleh2021blockchain}.
Our approach combines both, requiring arranger processes to stake
tokens when posting batch tags (see
Section~\ref{sec:incentives:stakes}) and rewarding arranger processes
based on their work when batch tags are confirmed
(Section~\ref{sec:incentives:generate}).
Other mechanisms can be explored to, for example, prevent
free-riders, including active insertion of invalid
blocks~\cite{liu2020game, smuseva2022verifier, alharby2020data} and
Randomness Inserted Contract Execution~\cite{das2018yoda}.

\subsubsection*{Refereed Delegation.}
%
The challenge mechanisms introduced in
Section~\ref{sec:incentives:challenges} involve bisection games which
has aspects that are very similar to Refereed Delegation of
Computation (RDoC).
RDoC is a two-server protocol for delegating
computation~\cite{canetti2011practical,canetti2013refereed} where
resource-bound clients outsource computation to powerful servers.
Servers have to provide proof-of-correct computations,
and verifying proofs should be more efficient than performing the
computation itself~\cite{goldwasser2015delegating}.
Other blockchain protocols also employ similar mechanisms to resolve disputes
and motivate correct
behavior~\cite{teutsch2019scalable,wagner2019dispute,nehab2022permissionless}.

%% file: main.bbl
\begin{thebibliography}{100}

\bibitem{linea}
{A Consensys Formation}.
\newblock Linea.
\newblock \url{https://linea.build/}, Accessed: 2024-05-21.

\bibitem{aevo}
{Aevo}.
\newblock \url{https://www.aevo.xyz/}, Accessed: 2024-05-21.

\bibitem{albassam2019fraud}
Mustafa Al-Bassam, Alberto Sonnino, and Vitalik Buterin.
\newblock Fraud and data availability proofs: Maximising light client security
  and scaling blockchains with dishonest majorities, 2019.
\newblock \href {https://arxiv.org/abs/1809.09044} {\path{arXiv:1809.09044}}.

\bibitem{Alakuijala18brotli}
Jyrki Alakuijala, Andrea Farruggia, Paolo Ferragina, Eugene Kliuchnikov, Robert
  Obryk, Zoltan Szabadka, and Lode Vandevenne.
\newblock Brotli: A general-purpose data compressor.
\newblock {\em ACM Transactions on Information Systems}, 37(1):1--30, 2018.

\bibitem{alharby2020data}
Maher Alharby, Roben~Castagna Lunardi, Amjad Aldweesh, and Aad Van~Moorsel.
\newblock Data-driven model-based analysis of the ethereum verifier's dilemma.
\newblock In {\em 2020 50th Annual IEEE/IFIP International Conference on
  Dependable Systems and Networks (DSN)}, pages 209--220. IEEE, 2020.

\bibitem{ancient8}
{Ancient8 Foundation}.
\newblock {Ancient8}.
\newblock \url{https://ancient8.gg/}, Accessed: 2024-05-22.

\bibitem{apex}
{ApeX}.
\newblock \url{https://apex-pro.gitbook.io/apex-pro?lang=en-US}, Accessed:
  2024-05-21.

\bibitem{astar}
{Astar Network}.
\newblock {Astar zkEVM}.
\newblock \url{https://astar.network/astar2}, Accessed: 2024-05-21.

\bibitem{bentov2017instantaneous}
Iddo Bentov, Ranjit Kumaresan, and Andrew Miller.
\newblock Instantaneous decentralized poker.
\newblock In {\em Advances in Cryptology--ASIACRYPT 2017: 23rd International
  Conference on the Theory and Applications of Cryptology and Information
  Security, Hong Kong, China, December 3-7, 2017, Proceedings, Part II 23},
  pages 410--440. Springer, 2017.

\bibitem{blast}
{Blast}.
\newblock \url{https://blast.io/en}, Accessed: 2024-05-21.

\bibitem{bob}
{BOB}.
\newblock \url{https://www.gobob.xyz/}, Accessed: 2024-05-21.

\bibitem{boba}
{Boba Network}.
\newblock \url{https://boba.network/}, Accessed: 2024-05-21.

\bibitem{Boneh2001Short}
Dan Boneh, Ben Lynn, and Hovav Shacham.
\newblock Short signatures from the weil pairing.
\newblock {\em Journal of Cryptology}, 17:297--319, 2001.

\bibitem{ArbitrumNitro}
L.~Bousfield, R.~Bousfield, C.~Buckland, B.~Burgess, J.~Colvin, E.~Felten,
  S.~Goldfeder, D.~Goldman, B.~Huddleston, H.~Kalonder, F.~Lacs, H.~Ng,
  A.~Sanghi, T.~Wilson, V.~Yermakova, and T.~Zidenberg.
\newblock Arbitrum nitro: A second-generation optimistic rollup, 2022.
\newblock
  \url{https://github.com/OffchainLabs/nitro/blob/master/docs/Nitro-whitepaper.pdf}.

\bibitem{brandenburger2018blockchain}
Marcus Brandenburger, Christian Cachin, R{\"u}diger Kapitza, and Alessandro
  Sorniotti.
\newblock Blockchain and trusted computing: Problems, pitfalls, and a solution
  for hyperledger fabric.
\newblock {\em arXiv preprint arXiv:1805.08541}, 2018.

\bibitem{buchman2016tendermint}
Ethan Buchman.
\newblock Tendermint: Byzantine fault tolerance in the age of blockchains.
\newblock Master's thesis, The University of Guelph, 2016.

\bibitem{cai2022benzene}
Zhongteng Cai, Junyuan Liang, Wuhui Chen, Zicong Hong, Hong-Ning Dai, Jianting
  Zhang, and Zibin Zheng.
\newblock Benzene: Scaling blockchain with cooperation-based sharding.
\newblock {\em IEEE Transactions on Parallel and Distributed Systems},
  34(2):639--654, 2022.

\bibitem{campanelli2017zkcontingent}
Matteo Campanelli, Rosario Gennaro, Steven Goldfeder, and Luca Nizzardo.
\newblock Zero-knowledge contingent payments revisited: Attacks and payments
  for services.
\newblock In {\em Proceedings of the 2017 ACM SIGSAC Conference on Computer and
  Communications Security}, CCS '17, page 229–243, New York, NY, USA, 2017.
  Association for Computing Machinery.
\newblock \href {https://doi.org/10.1145/3133956.3134060}
  {\path{doi:10.1145/3133956.3134060}}.

\bibitem{canetti2011practical}
Ran Canetti, Ben Riva, and Guy~N Rothblum.
\newblock Practical delegation of computation using multiple servers.
\newblock In {\em Proceedings of the 18th ACM conference on Computer and
  communications security}, pages 445--454, 2011.

\bibitem{canetti2013refereed}
Ran Canetti, Ben Riva, and Guy~N Rothblum.
\newblock Refereed delegation of computation.
\newblock {\em Information and Computation}, 226:16--36, 2013.

\bibitem{capretto22setchain}
Margarita Capretto, Mart\'in Ceresa, Antonio~Fern\'andez Anta, Antonio Russo,
  and C\'esar S\'anchez.
\newblock {S}etchain: Improving {B}lockchain scalability with {B}yzantine
  distributed sets and barriers.
\newblock In {\em Proc. of the 2022 IEEE International Conference on
  Blockchain}, pages 87--96. IEEE, 2022.
\newblock \href {https://doi.org/10.1109/Blockchain55522.2022.00022}
  {\path{doi:10.1109/Blockchain55522.2022.00022}}.

\bibitem{cheng2019ekiden}
Raymond Cheng, Fan Zhang, Jernej Kos, Warren He, Nicholas Hynes, Noah Johnson,
  Ari Juels, Andrew Miller, and Dawn Song.
\newblock Ekiden: A platform for confidentiality-preserving, trustworthy, and
  performant smart contracts.
\newblock In {\em 2019 IEEE European Symposium on Security and Privacy
  (EuroS\&P)}, pages 185--200. IEEE, 2019.

\bibitem{base}
{Coinbase}.
\newblock {base}.
\newblock \url{https://base.org/}, Accessed: 2024-05-21.

\bibitem{Coleman2015StateChannels}
Jeff Coleman.
\newblock State channels, 2015.
\newblock \url{https://www.jeffcoleman.ca/state-channels/}.

\bibitem{cometBFT}
CometBFT.
\newblock What is {CometBFT}.
\newblock \url{https://docs.cometbft.com/v0.37/introduction/}, Accessed:
  2024-05-21.

\bibitem{Crain2021RedBelly}
Tyler Crain, Christopher Natoli, and Vincent Gramoli.
\newblock Red belly: A secure, fair and scalable open blockchain.
\newblock In {\em Proc. of S\&P'21}, pages 466--483, 2021.
\newblock \href {https://doi.org/10.1109/SP40001.2021.00087}
  {\path{doi:10.1109/SP40001.2021.00087}}.

\bibitem{Cristin1996AtomicBroadcast}
F.~Cristian, H.~Aghili, R.~Strong, and D.~Volev.
\newblock Atomic broadcast: from simple message diffusion to byzantine
  agreement.
\newblock In {\em 25th Int'l Symp. on Fault-Tolerant Computing}, pages 431--,
  1995.
\newblock \href {https://doi.org/10.1109/FTCSH.1995.532668}
  {\path{doi:10.1109/FTCSH.1995.532668}}.

\bibitem{croman2016scaling}
Kyle Croman, Christian Decker, Ittay Eyal, Adem~Efe Gencer, Ari Juels, Ahmed
  Kosba, Andrew Miller, Prateek Saxena, Elaine Shi, Emin G{\"u}n~Sirer, et~al.
\newblock On scaling decentralized blockchains: (a position paper).
\newblock In {\em International conference on financial cryptography and data
  security}, pages 106--125. Springer, 2016.

\bibitem{cyber}
{Cyber}.
\newblock \url{https://cyber.co/}, Accessed: 2024-05-22.

\bibitem{Dang2019Sharding}
Hung Dang, Tien Tuan~Anh Dinh, Dumitrel Loghin, Ee-Chien Chang, Qian Lin, and
  Beng~Chin Ooi.
\newblock Towards scaling blockchain systems via sharding.
\newblock In {\em Proc. of SIGMOD'19}, pages 123–--140. ACM, 2019.
\newblock \href {https://doi.org/10.1145/3299869.3319889}
  {\path{doi:10.1145/3299869.3319889}}.

\bibitem{das2019fastkitten}
Poulami Das, Lisa Eckey, Tommaso Frassetto, David Gens, Kristina
  Host{\'a}kov{\'a}, Patrick Jauernig, Sebastian Faust, and Ahmad-Reza Sadeghi.
\newblock Fastkitten: Practical smart contracts on bitcoin.
\newblock In {\em 28th USENIX Security Symposium (USENIX Security 19)}, pages
  801--818, 2019.

\bibitem{das2020efficient}
Sourav Das, Vinith Krishnan, and Ling Ren.
\newblock Efficient cross-shard transaction execution in sharded blockchains.
\newblock {\em arXiv preprint arXiv:2007.14521}, 2020.

\bibitem{das2018yoda}
Sourav Das, Vinay~Joseph Ribeiro, and Abhijeet Anand.
\newblock Yoda: Enabling computationally intensive contracts on blockchains
  with byzantine and selfish nodes.
\newblock {\em arXiv preprint arXiv:1811.03265}, 2018.

\bibitem{degate}
{DeGate Home DAO}.
\newblock {DeGate V1}.
\newblock \url{https://degate.com/}, Accessed: 2024-05-21.

\bibitem{metisDecentralized}
Metis~Developer Documentation.
\newblock Decentralized sequencer.
\newblock \url{https://docs.metis.io/dev/decentralized-sequencer/overview},
  Accessed: 2024-05-21.

\bibitem{dwork1992pricing}
Cynthia Dwork and Moni Naor.
\newblock Pricing via processing or combatting junk mail.
\newblock In {\em Annual international cryptology conference}, pages 139--147.
  Springer, 1992.

\bibitem{dydx}
{dYdX}.
\newblock \url{https://dydx.exchange/}, Accessed: 2024-05-21.

\bibitem{Dziembowski2017PERUNVP}
Stefan Dziembowski, Lisa Eckey, Sebastian Faust, and Daniel Malinowski.
\newblock Perun: Virtual payment channels over cryptographic currencies.
\newblock {\em IACR Cryptol. ePrint Arch.}, 2017:635, 2017.
\newblock URL: \url{https://api.semanticscholar.org/CorpusID:26535415}.

\bibitem{Fischer1985Impossibility}
Michael~J. Fischer, Nancy~A. Lynch, and Michael~S. Paterson.
\newblock Impossibility of distributed consensus with one faulty process.
\newblock {\em JACM}, 32(2):374--382, 1985.
\newblock \href {https://doi.org/10.1145/3149.214121}
  {\path{doi:10.1145/3149.214121}}.

\bibitem{frassetto2022pose}
Tommaso Frassetto, Patrick Jauernig, David Koisser, David Kretzler, Benjamin
  Schlosser, Sebastian Faust, and Ahmad-Reza Sadeghi.
\newblock Pose: Practical off-chain smart contract execution.
\newblock {\em arXiv preprint arXiv:2210.07110}, 2022.

\bibitem{fraxtal}
{Fraxtal}.
\newblock {Frax Finance}.
\newblock \url{https://www.frax.com/}, Accessed: 2024-05-21.

\bibitem{fuelv1}
{Fuel Labs}.
\newblock {Fuel V1}.
\newblock \url{https://fuel.network/}, Accessed: 2024-05-22.

\bibitem{goldwasser2015delegating}
Shafi Goldwasser, Yael~Tauman Kalai, and Guy~N Rothblum.
\newblock Delegating computation: interactive proofs for muggles.
\newblock {\em Journal of the ACM (JACM)}, 62(4):1--64, 2015.

\bibitem{gudgeon2020sok}
Lewis Gudgeon, Pedro Moreno-Sanchez, Stefanie Roos, Patrick McCorry, and Arthur
  Gervais.
\newblock Sok: Layer-two blockchain protocols.
\newblock In {\em Financial Cryptography and Data Security: 24th International
  Conference, FC 2020, Kota Kinabalu, Malaysia, February 10--14, 2020 Revised
  Selected Papers 24}, pages 201--226. Springer, 2020.

\bibitem{hafid2020scaling}
Abdelatif Hafid, Abdelhakim~Senhaji Hafid, and Mustapha Samih.
\newblock Scaling blockchains: A comprehensive survey.
\newblock {\em IEEE access}, 8:125244--125262, 2020.

\bibitem{hewa2021survey}
Tharaka Hewa, Mika Ylianttila, and Madhusanka Liyanage.
\newblock Survey on blockchain based smart contracts: Applications,
  opportunities and challenges.
\newblock {\em Journal of network and computer applications}, 177:102857, 2021.

\bibitem{hypr}
{Hypr Network}.
\newblock {Hypr}.
\newblock \url{https://www.hypr.network/}, Accessed: 2024-05-22.

\bibitem{immutablex}
{Immutable X}.
\newblock \url{https://www.immutable.com/}, Accessed: 2024-05-21.

\bibitem{jakobsson1999proofs}
Markus Jakobsson and Ari Juels.
\newblock Proofs of work and bread pudding protocols.
\newblock In {\em Secure Information Networks: Communications and Multimedia
  Security IFIP TC6/TC11 Joint Working Conference on Communications and
  Multimedia Security (CMS’99) September 20--21, 1999, Leuven, Belgium},
  pages 258--272. Springer, 1999.

\bibitem{jia2022themis}
Linpeng Jia, Keyuan Wang, Xin Wang, Lei Yu, Zhongcheng Li, and Yi~Sun.
\newblock Themis: An equal, unpredictable, and scalable consensus for
  consortium blockchain.
\newblock In {\em 2022 IEEE 42nd International Conference on Distributed
  Computing Systems (ICDCS)}, pages 235--245. IEEE, 2022.

\bibitem{Kalodner2018Arbitrum}
Harry Kalodner, Steven Goldfeder, Xiaoqi Chen, S.~Matthew Weinberg, and
  Edward~W. Felten.
\newblock Arbitrum: Scalable, private smart contracts.
\newblock In {\em 27th {USENIX} Security Symposium}, pages 1353--1370. {USENIX}
  Assoc., 2018.
\newblock URL:
  \url{https://www.usenix.org/conference/usenixsecurity18/presentation/kalodner}.

\bibitem{karak}
{Karak}.
\newblock \url{https://karak.network/}, Accessed: 2024-05-21.

\bibitem{kelkar2020orderfairness}
Mahimna Kelkar, Fan Zhang, Steven Goldfeder, and Ari Juels.
\newblock Order-fairness for byzantine consensus.
\newblock Cryptology ePrint Archive, Paper 2020/269, 2020.
\newblock \url{https://eprint.iacr.org/2020/269}.
\newblock URL: \url{https://eprint.iacr.org/2020/269}.

\bibitem{kinto}
{Kinto}.
\newblock \url{https://www.kinto.xyz/}, Accessed: 2024-05-21.

\bibitem{kokoris2018omniledger}
Eleftherios Kokoris-Kogias, Philipp Jovanovic, Linus Gasser, Nicolas Gailly,
  Ewa Syta, and Bryan Ford.
\newblock Omniledger: A secure, scale-out, decentralized ledger via sharding.
\newblock In {\em 2018 IEEE symposium on security and privacy (SP)}, pages
  583--598. IEEE, 2018.

\bibitem{kwon2019cosmos}
Jae Kwon and Ethan Buchman.
\newblock Cosmos whitepaper, 2019.
\newblock \url{https://cosmos. network/resources/whitepaper}.

\bibitem{celestia}
Celestia Labs.
\newblock Celestia.
\newblock \url{https://celestia.org}, Accessed: 2024-05-21.

\bibitem{ArbitrumNitroGithub}
Offchain Labs.
\newblock Arbitrum nitro, 2022.
\newblock \url{https://github.com/OffchainLabs/nitro}.

\bibitem{redstone}
{Lattice}.
\newblock {Redstone}.
\newblock \url{https://redstone.xyz/}, Accessed: 2024-05-22.

\bibitem{kroma}
{Lightscale Holdings}.
\newblock {Kroma}.
\newblock \url{https://kroma.network/}, Accessed: 2024-05-21.

\bibitem{lind2019teechain}
Joshua Lind, Ittay Eyal, Florian Kelbert, Oded Naor, Peter~R. Pietzuch, and
  Emin~G{\"{u}}n Sirer.
\newblock Teechain: Scalable blockchain payments using trusted execution
  environments.
\newblock {\em CoRR}, abs/1707.05454, 2017.
\newblock URL: \url{http://arxiv.org/abs/1707.05454}, \href
  {https://arxiv.org/abs/1707.05454} {\path{arXiv:1707.05454}}.

\bibitem{liu2021parallel}
Jian Liu, Peilun Li, Raymond Cheng, N~Asokan, and Dawn Song.
\newblock Parallel and asynchronous smart contract execution.
\newblock {\em IEEE Transactions on Parallel and Distributed Systems},
  33(5):1097--1108, 2021.

\bibitem{liu2020game}
Pinglan Liu and Wensheng Zhang.
\newblock Game theoretic approach for secure and efficient heavy-duty smart
  contracts.
\newblock In {\em 2020 IEEE Conference on Communications and Network Security
  (CNS)}, pages 1--9. IEEE, 2020.

\bibitem{loopring}
{Loopring Project Ltd}.
\newblock {Loopring}.
\newblock \url{https://loopring.org}, Accessed: 2024-05-21.

\bibitem{lovejoyparsec}
James Lovejoy, Anders Brownworth, Madars Virza, and Neha Narula.
\newblock Parsec: Executing smart contracts in parallel.
\newblock \url{https://www.media.mit.edu/projects/parsec/overview/}.

\bibitem{lyra}
{Lyra Mobile}.
\newblock {Lyra}.
\newblock \url{https://lyra.finance/}, Accessed: 2024-05-21.

\bibitem{Zamani2018RapidChain}
Zamani Mahdi, Mahnush Movahedi, and Mariana Raykova.
\newblock Rapidchain: Scaling blockchain via full sharding.
\newblock In {\em Proc. of CSS'18}, pages 931–--948. ACM, 2018.
\newblock \href {https://doi.org/10.1145/3243734.3243853}
  {\path{doi:10.1145/3243734.3243853}}.

\bibitem{hotstuff}
Dahlia Malkhi and Kartik Nayak.
\newblock Extended abstract: Hotstuff-2: Optimal two-phase responsive bft.
\newblock Cryptology ePrint Archive, Paper 2023/397, 2023.
\newblock \url{https://eprint.iacr.org/2023/397}.
\newblock URL: \url{https://eprint.iacr.org/2023/397}.

\bibitem{mallaki2022off}
Mahdi Mallaki, Babak Majidi, Amirhossein Peyvandi, and Ali Movaghar.
\newblock Off-chain management and state-tracking of smart programs on
  blockchain for secure and efficient decentralized computation.
\newblock {\em International Journal of Computers and Applications},
  44(9):822--829, 2022.

\bibitem{mantle}
{Mantle}.
\newblock \url{https://www.mantle.xyz/}, Accessed: 2024-05-21.

\bibitem{zksyncera}
{Matter Labs}.
\newblock {zkSync}.
\newblock \url{https://zksync.io/}, Accessed: 2024-05-21.

\bibitem{zksynclite}
{Matter Labs}.
\newblock {zkSync Lite}.
\newblock \url{https://zksync.io/}, Accessed: 2024-05-21.

\bibitem{Merkle88}
Ralph~C. Merkle.
\newblock A digital signature based on a conventional encryption function.
\newblock In Carl Pomerance, editor, {\em Advances in Cryptology --- CRYPTO
  '87}, pages 369--378, Berlin, Heidelberg, 1988. Springer Berlin Heidelberg.

\bibitem{metal}
{Metallicus}.
\newblock {Metal L2}.
\newblock \url{https://metall2.com/}, Accessed: 2024-05-22.

\bibitem{metis}
{Metis}.
\newblock {Metis Andromeda}.
\newblock \url{https://www.metis.io/}, Accessed: 2024-05-21.

\bibitem{miller2017sprites}
Andrew Miller, Iddo Bentov, Ranjit Kumaresan, and Patrick McCorry.
\newblock Sprites: Payment channels that go faster than lightning.
\newblock {\em CoRR}, abs/1702.05812, 2017.
\newblock URL: \url{http://arxiv.org/abs/1702.05812}, \href
  {https://arxiv.org/abs/1702.05812} {\path{arXiv:1702.05812}}.

\bibitem{mode}
{Mode Network}.
\newblock \url{https://www.mode.network/}, Accessed: 2024-05-21.

\bibitem{motepalli2023sok}
Shashank Motepalli, Luciano Freitas, and Benjamin Livshits.
\newblock Sok: Decentralized sequencers for rollups, 2023.
\newblock \href {https://arxiv.org/abs/2310.03616} {\path{arXiv:2310.03616}}.

\bibitem{myria}
{Myria}.
\newblock \url{https://myria.com/}, Accessed: 2024-05-22.

\bibitem{Astria}
Eshita Nandini.
\newblock Introducing the astria development cluster.
\newblock
  \url{https://www.astria.org/blog/introducing-the-astria-development-cluster},
  Accessed: 2024-05-21.

\bibitem{nasir2022scalable}
Muhammad~Hassan Nasir, Junaid Arshad, Muhammad~Mubashir Khan, Mahawish Fatima,
  Khaled Salah, and Raja Jayaraman.
\newblock Scalable blockchains—a systematic review.
\newblock {\em Future generation computer systems}, 126:136--162, 2022.

\bibitem{nehab2022permissionless}
Diego Nehab and Augusto Teixeira.
\newblock Permissionless refereed tournaments, 2022.
\newblock \href {https://arxiv.org/abs/2212.12439} {\path{arXiv:2212.12439}}.

\bibitem{nguyen2020WIIsAlmostEnough}
Ky~Nguyen, Miguel Ambrona, and Masayuki Abe.
\newblock Wi is almost enough: Contingent payment all over again.
\newblock In {\em Proceedings of the 2020 ACM SIGSAC Conference on Computer and
  Communications Security}, CCS '20, page 641–656, New York, NY, USA, 2020.
  Association for Computing Machinery.
\newblock \href {https://doi.org/10.1145/3372297.3417888}
  {\path{doi:10.1145/3372297.3417888}}.

\bibitem{nick2020liquid}
Jonas Nick, Andrew Poelstra, and Gregory Sanders.
\newblock Liquid: A bitcoin sidechain.
\newblock {\em Liquid white paper. URL https://blockstream.
  com/assets/downloads/pdf/liquid-whitepaper. pdf}, 2020.

\bibitem{xlayer}
{OKX}.
\newblock {X Layer}.
\newblock \url{https://www.okx.com/xlayer}, Accessed: 2024-05-21.

\bibitem{ongaro2014raft}
Diego Ongaro and John Ousterhout.
\newblock In search of an understandable consensus algorithm.
\newblock In {\em 2014 USENIX Annual Technical Conference (USENIX ATC 14)},
  pages 305--319, Philadelphia, PA, June 2014. USENIX Association.
\newblock URL:
  \url{https://www.usenix.org/conference/atc14/technical-sessions/presentation/ongaro}.

\bibitem{optimism}
{Optimism Foundation}.
\newblock {Optimism}.
\newblock \url{https://www.optimism.io/}, Accessed: 2024-05-21.

\bibitem{orderly}
{Orderly Network}.
\newblock \url{https://orderly.network/}, Accessed: 2024-05-22.

\bibitem{manta}
{p0x Labs}.
\newblock {Manta Pacific}.
\newblock \url{https://pacific.manta.network/}, Accessed: 2024-05-21.

\bibitem{paradex}
{Paradex}.
\newblock \url{https://www.paradex.trade/}, Accessed: 2024-05-21.

\bibitem{parallel}
{Parallel}.
\newblock \url{https://parallel.fi/}, Accessed: 2024-05-22.

\bibitem{pgn}
{Public Goods Network}.
\newblock \url{https://publicgoods.network/}, Accessed: 2024-05-22.

\bibitem{polygonPoS}
{Polygon Labs UI}.
\newblock {Polygon PoS}.
\newblock \url{https://polygon.technology/polygon-pos}.

\bibitem{polygonzkEVM}
{Polygon Labs UI}.
\newblock {Polygon zkEVM}.
\newblock \url{https://polygon.technology/polygon-zkevm}, Accessed: 2024-05-21.

\bibitem{poon2017plasma}
Joseph Poon and Vitalik Buterin.
\newblock Plasma: Scalable autonomous smart contracts, 2017.
\newblock \url{https://plasma.io/}.

\bibitem{Poon2016lightning}
Joseph Poon and Thaddeus Dryja.
\newblock The bitcoin lightning network: Scalable off-chain instant payments,
  2016.
\newblock \url{https://lightning.network/lightning-network-paper.pdf}.

\bibitem{Radius}
Radius.
\newblock Introduction to radius.
\newblock \url{https://docs.theradius.xyz/overview/introduction-to-radius},
  Accessed: 2024-05-21.

\bibitem{rana2022free2shard}
Ranvir Rana, Sreeram Kannan, David Tse, and Pramod Viswanath.
\newblock Free2shard: Adversary-resistant distributed resource allocation for
  blockchains.
\newblock {\em Proceedings of the ACM on Measurement and Analysis of Computing
  Systems}, 6(1):1--38, 2022.

\bibitem{reddioex}
{Reddio}.
\newblock \url{https://www.reddio.com/}, Accessed: 2024-05-22.

\bibitem{l2beat}
L2BEAT research team.
\newblock L2beat.
\newblock Accessed: 2024-05-21.
\newblock URL: \url{https://l2beat.com/scaling/summary}.

\bibitem{reya}
{Reya}.
\newblock \url{https://reya.network/}, Accessed: 2024-05-21.

\bibitem{rhinofi}
{rhino.fi}.
\newblock \url{https://rhino.fi/}, Accessed: 2024-05-21.

\bibitem{rss3}
{RSS3}.
\newblock \url{https://rss3.io/}, Accessed: 2024-05-21.

\bibitem{saleh2021blockchain}
Fahad Saleh.
\newblock Blockchain without waste: Proof-of-stake.
\newblock {\em The Review of financial studies}, 34(3):1156--1190, 2021.

\bibitem{sanka2021systematic}
Abdurrashid~Ibrahim Sanka and Ray~CC Cheung.
\newblock A systematic review of blockchain scalability: Issues, solutions,
  analysis and future research.
\newblock {\em Journal of Network and Computer Applications}, 195:103232, 2021.

\bibitem{scroll}
{Scroll Ltd}.
\newblock {Scroll}.
\newblock \url{https://scroll.io/}, Accessed: 2024-05-21.

\bibitem{espressoSequencer}
Espresso Sequencer.
\newblock The espresso sequencer: Hotshot consensus and tiramisu data
  availability.
\newblock \url{https://webassembly.org/}, Accessed: 2024-05-21.
\newblock URL:
  \url{https://github.com/EspressoSystems/HotShot/blob/main/docs/espresso-sequencer-paper.pdf}.

\bibitem{shen2023gt}
Tao Shen, Tianyu Li, Zhuo Yu, Fenhua Bai, and Chi Zhang.
\newblock Gt-nrsm: efficient and scalable sharding consensus mechanism for
  consortium blockchain.
\newblock {\em The Journal of Supercomputing}, pages 1--35, 2023.

\bibitem{smuseva2022verifier}
Daria Smuseva, Ivan Malakhov, Andrea Marin, Aad van Moorsel, and Sabina Rossi.
\newblock Verifier’s dilemma in ethereum blockchain: A quantitative analysis.
\newblock In {\em International Conference on Quantitative Evaluation of
  Systems}, pages 317--336. Springer, 2022.

\bibitem{sorare}
{Sorare, SAS}.
\newblock {Sorare}.
\newblock \url{https://sorare.com/es/}, Accessed: 2024-05-21.

\bibitem{starknet}
{Starknet}.
\newblock \url{https://www.starknet.io/en}, Accessed: 2024-05-21.

\bibitem{tanX}
{tanX}.
\newblock \url{https://tanx.fi/}, Accessed: 2024-05-22.

\bibitem{boneh}
Ertem~Nusret Tas and Dan Boneh.
\newblock Cryptoeconomic security for data availability committees.
\newblock In Foteini Baldimtsi and Christian Cachin, editors, {\em Financial
  Cryptography and Data Security}, pages 310--326, Cham, 2024. Springer Nature
  Switzerland.

\bibitem{teutsch2019scalable}
Jason Teutsch and Christian Reitwie{\ss}ner.
\newblock A scalable verification solution for blockchains.
\newblock {\em CoRR}, abs/1908.04756, 2019.
\newblock URL: \url{http://arxiv.org/abs/1908.04756}, \href
  {https://arxiv.org/abs/1908.04756} {\path{arXiv:1908.04756}}.

\bibitem{cartesi}
{The Cartesi Foundation}.
\newblock {Cartesi}.
\newblock \url{https://cartesi.io/}, Accessed: 2024-05-22.

\bibitem{thibault2022blockchain}
Louis~Tremblay Thibault, Tom Sarry, and Abdelhakim~Senhaji Hafid.
\newblock Blockchain scaling using rollups: A comprehensive survey.
\newblock {\em IEEE Access}, 2022.

\bibitem{unnikrishnan2022survey}
KN~Unnikrishnan and P~Victer~Paul.
\newblock A survey on layer 2 solutions to achieve scalability in blockchain.
\newblock In {\em Advances in Distributed Computing and Machine Learning:
  Proceedings of ICADCML 2022}, pages 205--216. Springer, 2022.

\bibitem{fuchsbauer2019WIIsNotEnough}
Georg~Fuchsbauer \url{https://www.fuel.network/}.
\newblock Wi is not enough: Zero-knowledge contingent (service) payments
  revisited.
\newblock Cryptology ePrint Archive, Paper 2019/964, 2019.
\newblock \url{https://eprint.iacr.org/2019/964}.
\newblock URL: \url{https://eprint.iacr.org/2019/964}, \href
  {https://doi.org/10.1145/3319535.3354234}
  {\path{doi:10.1145/3319535.3354234}}.

\bibitem{wagner2019dispute}
Eric Wagner, Achim V{\"o}lker, Frederik Fuhrmann, Roman Matzutt, and Klaus
  Wehrle.
\newblock Dispute resolution for smart contract-based two-party protocols.
\newblock In {\em 2019 IEEE International Conference on Blockchain and
  Cryptocurrency (ICBC)}, pages 422--430. IEEE, 2019.

\bibitem{wang2020prism}
Gerui Wang, Shuo Wang, Vivek Bagaria, David Tse, and Pramod Viswanath.
\newblock Prism removes consensus bottleneck for smart contracts.
\newblock In {\em 2020 Crypto Valley Conference on Blockchain Technology
  (CVCBT)}, pages 68--77. IEEE, 2020.

\bibitem{Wang2019FastChain}
Ke~Wang and Hyong~S. Kim.
\newblock Fastchain: Scaling blockchain system with informed neighbor
  selection.
\newblock In {\em Proc. of IEEE Blockchain'19}, pages 376--383, 2019.
\newblock \href {https://doi.org/10.1109/Blockchain.2019.00058}
  {\path{doi:10.1109/Blockchain.2019.00058}}.

\bibitem{wasm}
{WebAssembly Working Group}.
\newblock {WebAssembly}.
\newblock \url{https://webassembly.org/}.

\bibitem{wood2014ethereum}
Gavin Wood.
\newblock Ethereum: A secure decentralised generalised transaction ledger.
\newblock {\em Ethereum project yellow paper}, 151:1--32, 2014.

\bibitem{wood2016polkadot}
Gavin Wood.
\newblock Polkadot: Vision for a heterogeneous multi-chain framework.
\newblock {\em White Paper}, 21, 2016.

\bibitem{xiao2020privacyguard}
Yang Xiao, Ning Zhang, Jin Li, Wenjing Lou, and Y~Thomas Hou.
\newblock Privacyguard: Enforcing private data usage control with blockchain
  and attested off-chain contract execution.
\newblock In {\em Computer Security--ESORICS 2020: 25th European Symposium on
  Research in Computer Security, ESORICS 2020, Guildford, UK, September 14--18,
  2020, Proceedings, Part II 25}, pages 610--629. Springer, 2020.

\bibitem{zamyatin2019xclaim}
Alexei Zamyatin, Dominik Harz, Joshua Lind, Panayiotis Panayiotou, Arthur
  Gervais, and William Knottenbelt.
\newblock Xclaim: Trustless, interoperable, cryptocurrency-backed assets.
\newblock In {\em 2019 IEEE Symposium on Security and Privacy (SP)}, pages
  193--210. IEEE, 2019.

\bibitem{zhang2019paralysis}
Fan Zhang, Philip Daian, Iddo Bentov, Ian Miers, and Ari Juels.
\newblock Paralysis proofs: Secure dynamic access structures for cryptocurrency
  custody and more.
\newblock In {\em Proceedings of the 1st ACM Conference on Advances in
  Financial Technologies}, pages 1--15, 2019.

\bibitem{zhou2020solutions}
Qiheng Zhou, Huawei Huang, Zibin Zheng, and Jing Bian.
\newblock Solutions to scalability of blockchain: A survey.
\newblock {\em IEEE Access}, 8:16440--16455, 2020.

\bibitem{zhou2023overview}
Qiheng Zhou, Huawei Huang, Zibin Zheng, and Jing Bian.
\newblock Overview to blockchain scalability challenges and solutions.
\newblock In {\em Blockchain Scalability}, pages 51--80. Springer, 2023.

\bibitem{zkfair}
{ZKFair Team}.
\newblock {ZKFair}.
\newblock \url{https://zkfair.io/}, Accessed: 2024-05-21.

\bibitem{zkspace}
{ZKSpace Team}.
\newblock {ZKSpace}.
\newblock \url{https://en.wiki.zks.org/}, Accessed: 2024-05-22.

\bibitem{zora}
{Zora Labs}.
\newblock {Zora}.
\newblock \url{https://docs.zora.co/docs/zora-network/intro}, Accessed:
  2024-05-21.

\end{thebibliography}
